\documentclass{amsart}
\usepackage{amsmath, amsthm, amsfonts, amssymb} 
\usepackage{mathrsfs}
\usepackage{dsfont}
\numberwithin{equation}{section}
\usepackage{a4}
\usepackage{url}
\usepackage{enumerate}
\usepackage[english]{babel}
\usepackage{graphicx}
\usepackage{stmaryrd}
\usepackage[authoryear]{natbib}
\usepackage[usenames, dvipsnames]{color}
\usepackage{hyperref}

\newtheorem{Theorem}{Theorem}[section]
\newtheorem*{Metatheorem}{Metatheorem}

\newtheorem{Prop}[Theorem]{Proposition}
\newtheorem{Lem}[Theorem]{Lemma}
\newtheorem{Ass}[Theorem]{Assumption}

\theoremstyle{remark}
\newtheorem{Rem}[Theorem]{Remark}
\newtheorem{Ex}[Theorem]{Example}

\theoremstyle{definition}
\newtheorem{Defn}[Theorem]{Definition}

\newcommand{\coloneqq}{\mathrel{\mathop:}=}

\renewcommand{\epsilon}{\varepsilon}

\newcommand{\R}{\mathds{R}}
\newcommand{\Q}{\mathds{Q}}

\newcommand{\N}{\mathds{N}}

\newcommand{\E}{\mathbb{E}}
\newcommand{\F}{\mathcal{F}}

\renewcommand{\P}{\mathds{P}}

\newcommand{\Qc}{\mathcal{Q}}
\newcommand{\Ac}{\mathcal{A}}

\newcommand{\Pc}{\mathcal{P}}
\newcommand{\Nc}{\mathcal{N}}
\newcommand{\Fc}{\mathcal{F}}
\newcommand{\Fb}{\mathbb{F}}
\newcommand{\Nb}{\mathbb{N}}
\newcommand{\Sc}{\mathcal{S}}
\newcommand{\Uc}{\mathcal{U}}
\newcommand{\Pfr}{\mathfrak{P}}
\newcommand{\nf}{\mathrm{0}}

\newcommand{\adm}{\Ac_\Phi(\Fb)}
\newcommand{\admU}{\Ac_\Phi(\Fb^\Uc)}

\newtheoremstyle{TheoremNum}
        {\topsep}{\topsep}              
        {\itshape}                      
        {}                              
        {\bfseries}                     
        {.}                             
        { }                             
        {\thmname{#1}\thmnote{ \bfseries #3}}
\theoremstyle{TheoremNum}

\usepackage{tikz}
\usetikzlibrary{decorations.pathreplacing}
\usepackage{tikz-3dplot}

\usepackage{tikz}
\usetikzlibrary{decorations.markings}
\usetikzlibrary{shapes.geometric}

\pgfdeclarelayer{edgelayer}
\pgfdeclarelayer{nodelayer}
\pgfsetlayers{edgelayer,nodelayer,main}

\tikzstyle{none}=[inner sep=0pt]
\tikzset{new/.style={thick}}

\definecolor{cccccc}{rgb}{0.8,0.8,0.8}
\definecolor{qqqqff}{rgb}{0,0,1}
\definecolor{cqcqcq}{rgb}{0.75,0.75,0.75}
\definecolor{ao}{rgb}{0.0, 0.5, 0.0}
\definecolor{amber}{rgb}{1.0, 0.0, 0.0}
\definecolor{babyblue}{rgb}{0.54, 0.81, 0.94}
\definecolor{red}{rgb}{1,0,0}
\definecolor{blue}{rgb}{0,0,1}
\definecolor{qreen}{rgb}{0,1,0}

\begin{document}

\title[Robust Modelling of Financial Markets]{A Unified Framework for Robust Modelling of Financial Markets in discrete time}

\author{Jan Ob{\l}{\'o}j}
\thanks{We gratefully acknowledge funding received from the European Research Council under the European Union's Seventh Framework Programme (FP7/2007-2013) / ERC grant agreement no. 335421) and from St. John's College in Oxford. We also thank Matteo Burzoni for his helpful remarks and comments. JW further acknowledges support from the German Academic Scholarship Foundation. }
\author{Johannes Wiesel}
\email{jan.obloj@maths.ox.ac.uk}
\email{johannes.wiesel@maths.ox.ac.uk}
\address{Mathematical Institute and St. John's College, University of Oxford, Oxford}

\date{\today}

\begin{abstract}
We unify and establish equivalence between the pathwise and the quasi-sure approaches to robust modelling of financial markets in discrete time. In particular, we prove a Fundamental Theorem of Asset Pricing and a Superhedging Theorem, which encompass the formulations of \citep{bouchard2015arbitrage}
and \citep{bfhmo}. 
In bringing the two streams of literature together, we also examine and relate their many different notions of arbitrage. We also clarify the relation between robust and classical $\P$-specific results. 
Furthermore, we prove when a superhedging property w.r.t.\ the set of martingale measures supported on a set of paths $\Omega$ may be extended to a pathwise superhedging on $\Omega$ without changing the superhedging price.
\end{abstract}

\maketitle 


\section{Introduction}
Mathematical models of financial markets are of great significance in economics and finance and have played a key role in the theory of pricing and hedging of derivatives and of risk management. Classical models, going back to \citep{samuelson1965rational} and \citep{black1973pricing} in continuous time, specify a fixed probability measure $\P$ to describe the asset price dynamics. They led to a powerful theory of complete, and later incomplete, financial markets. The original models have undergone a myriad of variations including, amongst others, local and stochastic volatility models and have been widely applied. However, they also faced important criticism for ignoring the issue of model uncertainty, particularly so in the wake of the 2007/08 financial crisis. Consequently, inspired by the theoretical developments going back to \citep{knight2012risk}, new modelling approaches emerged which aim to address this fundamental issue. These can be broadly divided into two streams based on the so-called quasi-sure and pathwise approaches respectively.

The quasi-sure approach introduces a set of priors $\mathfrak{P}$ representing possible market scenarios. These priors can be very different and $\mathfrak{P}$ typically contains measures which are mutually singular. This presents significant mathematical challenges and led to the theory of quasi-sure stochastic analysis (see, e.g., \citep{peng2004nonlinear,denis2006theoretical}). In discrete time, this framework was abstracted in \citep{bouchard2015arbitrage}, which we call the quasi-sure formulation in the rest of this paper. By varying the set of probability measures $\mathfrak{P}$ between the ``extreme" cases of one fixed probability measure, $\mathfrak{P}=\{\P\}$, and that of considering all probability measures, $\mathfrak{P}=\mathcal{P}(X)$, this formulation allows for widely different  specifications of market dynamics. The quasi-sure approach has been employed to consider model uncertainty along market frictions and other related problems, see e.g. \citep{bayraktar2015arbitrage,bayraktar2016fundamental}. 
The pathwise approach addresses Knightian uncertainty in market modelling by describing the set of market scenarios in absence of a probability measure or any similar relative weighting of such scenarios. 
It is also referred to as the pointwise, or $\omega$ by $\omega$, approach and it bears similarity to the way central banks carry out stress tests using scenario generators. In discrete time a suitable theory was obtained in \citep{bfhmo}, based on earlier developments in \citep{burzoni2015model,burzoni2016universal}. The methodology builds on the notion of prediction sets introduced in \citep{mykland2003financial} and used in continuous time in \citep{hou2015robust}. The particular case of including all scenarios is often referred to as the model-independent framework and was pioneered in \citep{davis2007range} and \citep{acciaio2013model}. From here, a further model specification is carried out by including additional assumptions, which represent the different agents' beliefs. In this manner paths deemed impossible by all agents are eliminated. The remaining set of paths is then called the prediction set, or the model.

Both approaches, the quasi-sure and the pathwise, allow thus to interpolate between the two ends of the modelling spectrum, as identified by \citep{Merton73}: the model-independent and the model-specific settings (see Figure \ref{fig. 1}). In doing so, they allow to capture how their outputs change in function of adding or removing modelling assumptions, thus allowing to quantify the impact and risk that a given set of assumptions bear on the problem at hand, see \citep{Cont2006}.
\begin{figure}[h] 
	\centering
		\includegraphics[width=10cm]{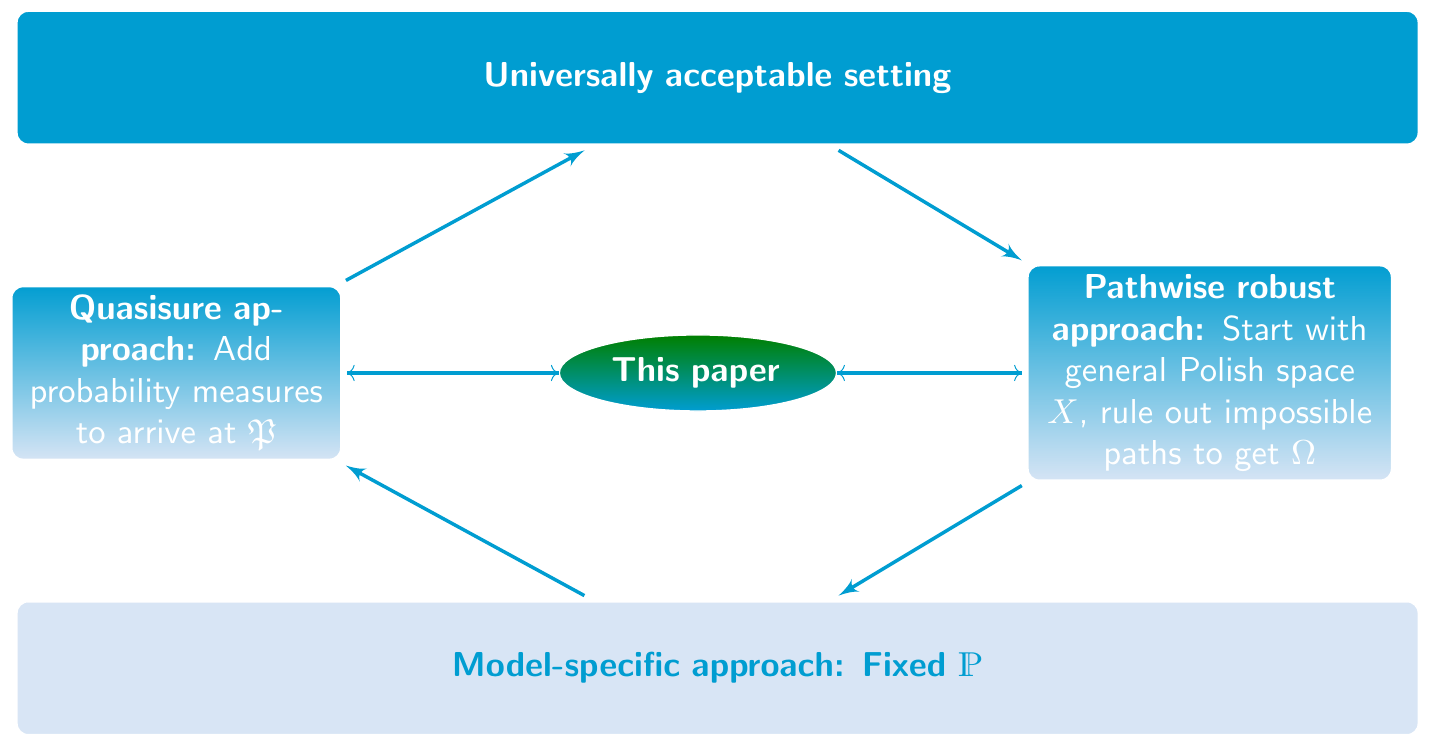}
	\caption[Different approaches to modelling financial markets]{Different approaches to modelling financial markets}
\label{fig. 1}
\end{figure}
Both approaches were successful in developing suitable notions of arbitrage and extending the core results from the classical $\P$-a.s. setting to their more general context. In particular, in both approaches, it is possible to establish a Fundamental Theorem of Asset pricing of the form
\begin{align*}
\text{No Arbitrage} \Leftrightarrow \text{Existence of martingale measures } \Q
\end{align*}
and a Superhedging Theorem of the form
\begin{align*}
\sup_{\Q} \E_{\Q}[g]= \inf \{x \ | \ x\text{ is the initial capital of a superhedging strategy of }g \}.
\end{align*}
Our main contribution is to unify these two approaches to model uncertainty. We show that, under mild technical assumptions, the pathwise and quasi-sure Fundamental Theorems of Asset Pricing and Superhedging Dualities can be inferred from one another and are thus equivalent. Our statements follow a meta-structure outlined below:
\begin{Metatheorem}
Suppose we are in the quasi-sure setting with a given set of priors $\mathfrak{P}$. Then, there exists a suitable selection of scenarios $\Omega^\mathfrak{P}$ such that the pathwise result for $\Omega^\mathfrak{P}$ implies the quasi-sure result for $\mathfrak{P}$.\\
Conversely, suppose we are given a selection of scenarios $\Omega$. Then, there is a set of priors $\mathfrak{P}^\Omega$ such that the quasi-sure result for $\mathfrak{P}^\Omega$ implies the pathwise result for $\Omega$.
\end{Metatheorem}

Establishing such equivalence allows us to gain significant additional insights into the core objects in both approaches, as well as clarify links to the classical model-specific setting. In particular, when transposing the results from the pointwise to the quasi-sure setup, the key technical \emph{analytic product structure} assumption in \cite{bouchard2015arbitrage}, see Definition \ref{def:APS} below, is deduced naturally from the analyticity of the set of scenarios in \citep{bfhmo}. When establishing the Superhedging Theorem, we not only show that the pathwise superhedging price of $g$ is equal to the quasi-sure one, but we also show that both are equal to the model-specific $\P$-superhedging price, where $\P$ depends on the setting, i.e., on $\Pfr$ or equivalently on $\Omega$, but also on the payoff $g$. Finally, the key implication in the proof of the robust Fundamental Theorem of Asset Pricing, i.e., \textit{(5)} $\Rightarrow$ \textit{(1)} in Theorem \ref{Thm. S} below, is obtained by carefully constructing a suitable $\P\in \Pfr$ which does not admit an arbitrage in the classical sense and hence admits an equivalent martingale measure.

Furthermore, we survey and relate the concepts of arbitrage used in both approaches. We provide an extensive list of arbitrage notions introduced and used across the literature on robust finance and establish clear relations between them.  
We also investigate in detail the notion of pathwise superhedging. As noted in \citep{burzoni2015model}, the pathwise superhedging duality does not hold for general claims $g$ when superhedging on a general set $\Omega$ is required. Instead, one has to consider hedging on a smaller ``efficient" set $\Omega^*$ (defined as the largest set supported by martingale measures and contained in $\Omega$) to retain the pricing-hedging duality. We clarify when this is necessary and when one can extend the superhedging duality from $\Omega^*$ to $\Omega$. Intuitively, since there are arbitrage opportunities on $\Omega \setminus \Omega^*$, one could try to superhedge the claim $g$ on $\Omega \setminus \Omega^*$ without any additional cost by implementing an arbitrage strategy. We provide a number of counterexamples to show this idea is not feasible in general and link this to measurability constraints on arbitrage strategies, which were also encountered in \citep{burzoni2016universal}. We then show that the above-mentioned intuition is only true for essentially uniformly continuous $g$ under certain regularity conditions on $\Omega$. 

The rest of the paper is organised as follows. Section \ref{sec:main} contains the main results. 
First, in Section \ref{sec 1}, we introduce the general setup in which we work. We discuss different notions of (robust) arbitrage in Section \ref{sec:arbitrage}. Then, in Section \ref{sec:FTAP}, we establish our version of the robust Fundamental Theorem of Asset Pricing which unifies the quasi-sure and pathwise perspectives. And in Section \ref{subseq superh}, we state a robust Superhedging Theorem. Section \ref{sec:additional} presents complementary results on extending the superhedging duality from $\Omega^*$ to $\Omega$ without additional cost and on relations between two strong notions of pathwise arbitrage. Finally, Section \ref{sec:proofs} contains technical results and most of the proofs. In particular, we give the proofs of Theorems \ref{Thm. bn vs bfhmo} and \ref{Thm. S}  in Section \ref{subseq. ftap} and of Theorem \ref{Thm super} in Section \ref{Sec proofsuper}.

\section{Unified Framework for Robust Modelling of Financial Markets}\label{sec:main}
\subsection{Trading strategies and pricing measures}\label{sec 1}

We use notation similar to \citep{bouchard2015arbitrage} and work in their setting, so we only recall the main objects of interest here and refer to \citep{bouchard2015arbitrage} and \citep[Chapter 7]{bertsekas1978stochastic} for technical details. Let $T \in \N$ and $X_1$ be a Polish space. We define for $t \in \{1, \dots T\}$ the Cartesian product $X_t\coloneqq X_1^t$ and define $X\coloneqq X_T$, with the convention that $X_0$ is a singleton. We denote by $\mathcal{B}(X)$ the Borel sets on $X$, by $\mathcal{P}(X)$ the set of probability measures on $\mathcal{B}(X)$ and define the function $\text{proj}_t\colon X \to X_1$ which projects $\omega \in X$ to the $t$-th coordinate, i.e., $\text{proj}_t(\omega)=\omega_t$. \\
Next we specify the financial market. Let $d \in \N$, $\mathbb{F}$ an arbitrary filtration and let $S_t=(S_t^1, \dots, S_t^d)\colon X_t \to \R^d$ be Borel-measurable, $0\leq t \leq T$, and adapted. 
All prices are given in units of a numeraire, $S^0$, which itself is thus normalised, $S^0_t\equiv 1$, $0\leq t\leq T$. 
Trading strategies $\mathcal{H}(\mathbb{F})$ are defined as the set of $\mathbb{F}$-predictable $\R^d$-valued processes. All trading is frictionless and self-financing. Given $H \in \mathcal{H}(\mathbb{F})$, we denote
\begin{align*}
H \circ S_t= \sum_{u=1}^t H_u \Delta S_u
\end{align*}
with $H\circ S_t$ representing the cashflow at time $t$ from trading using $H$. 
Above, and throughout, $H$ is a row vector, $S$ is a column vector and $1$ denotes either a scalar or a column vector $(1,\ldots, 1)^{\mathrm{T}}$.
We let $\Phi$ denote the vector of payoffs of the statically traded assets $\Phi = (\phi_\lambda: \lambda\in \Lambda)$, where $\Lambda$ is some index set. For notational convenience, we often identify $\Phi$ with the set of its elements. We assume that each $\phi \in \Phi$ is Borel-measurable. When there are no statically traded assets we write $\Phi = 0$. These assets, which we think of as options, can only be bought or sold at time zero (without loss of generality at zero cost) and are held until maturity $T$. A trading position $h$ can only hold finitely many of these assets, $h\in c_{00}(\Lambda)$ the space of sequences of reals indexed by $\Lambda$ with only finitely many non-zero elements, and generates the payoff $h \cdot \Phi=\sum_{\lambda \in \Lambda} h_\lambda \phi_\lambda$ at time T. We call a pair $(h,H) \in c_{00}(\Lambda)\times\mathcal{H}(\mathbb{F})$ a semistatic trading strategy. The class of such strategies is denoted $\adm:= c_{00}(\Lambda)\times\mathcal{H}(\mathbb{F})$.
For technical reasons we also introduce the level sets of $S$, which are denoted by
\begin{align*}
\Sigma_t^{\omega}= \{ \tilde{\omega} \in X \ | \ S_{0:t}(\omega)=S_{0:t}(\tilde{\omega}) \}
\end{align*}
for $t\in \{0, \dots, T\}$ and $\omega \in X_t$, where $S_{0:t} \coloneqq (S_0, \dots S_t)$. 
Finally, we denote by $\mathbb{F}^\nf=(\F_t^\nf)_{t=0, \dots, T}$ the natural filtration generated by $S$ and let $\F_t^{\mathcal{U}}$ be the universal completion of $\F_t^\nf$, $t=0, \dots, T$. Furthermore we write $(X, \mathcal{F}^{\mathcal{U}})$ for $(X_T , \mathcal{F}_T^{\mathcal{U}})$ and often consider $(X_t, \mathcal{F}_t^{\mathcal{U}})$ as a subspace of $(X, \mathcal{F}^{\mathcal{U}})$.\\
Within this setup, the literature on robust pricing and hedging adopts two approaches to model an agent's beliefs. One stream is scenario-based and proceeds by specifying a prediction set $\Omega \subseteq X$, which describes the possible price trajectories. The other stream proceeds by specifying a set  of probability measures $\mathfrak{P} \subseteq \mathcal{P}(X)$, which determines the set of negligible outcomes. 
We refer to the latter as the quasi-sure approach, while the former is usually called the pathwise, or pointwise, approach. In both cases, the model specification may depend on the agent's market information as well as on her specific modelling assumptions. Changing the sets $\Omega$ or $\mathfrak{P}$ can be seen as a natural way to interpolate between different beliefs. One of the principal aims of this paper is to show that both model approaches are equivalent in terms of corresponding FTAPs and Superhedging prices.\\

In order to aggregate trading strategies on different level sets $\Sigma_t^{\omega}$ in a measurable way, we always assume in this paper that $\Omega$ is analytic and $\mathfrak{P}$ has the following structure:

\begin{Defn}\label{def:APS}
A set $\mathfrak{P}\subseteq \mathcal{P}(X)$ is said to satisfy the Analytic Product Structure condition (APS), if
\begin{align*}
\mathfrak{P}=\{\P_0 \otimes \cdots \otimes \P_{T-1} \ | \ \P_t \text{ is}\  \mathcal{F}^{\mathcal{U}}_t\text{-measurable selector of }\mathfrak{P}_t\},
\end{align*}
where the sets $\mathfrak{P}_t(\omega) \subseteq \mathcal{P}(X_1)$ are nonempty, convex and 
\begin{align*}
\text{graph}(\mathfrak{P}_t) = \{ (\omega, \P)\ | \ \omega \in X_t, \ \P \in \mathfrak{P}_t (\omega)\}
\end{align*}
is analytic.
\end{Defn}
This structure facilitates a dynamic programming principle and allows to essentially paste together one-step results in order to establish their multistep counterparts.\\
In order to formulate a Fundamental Theorem of Asset pricing we need to define the dual objects to trading strategies: the pricing (martingale) measures. Given a set of measures $\mathfrak{P}$, following \citep{bouchard2015arbitrage}, we define
\begin{align*}
\Qc_{\Pfr,\Phi}\coloneqq \{ &\Q \in \mathcal{P}(X)\ | \ S \text{ is an } \mathbb{F}^{\mathcal{U}}\text{-martingale under }\Q, \  \exists \P \in \mathfrak{P} \text{ s.t. } \Q \ll \P, \\ 
&\ \E_{\Q}[\phi]=0 \ \forall \phi\in \Phi \},
\end{align*}
which, in the model-specific case $\mathfrak{P}=\{\P\}$, is simply the familiar set of all martingale measures equivalent to $\P$. 
Within the pathwise approach, for a set $\Omega\subseteq X$ and a filtration $\mathbb{F}$, we define
\begin{align*}
\mathcal{M}^f_{\Omega, \Phi}(\mathbb{F})\coloneqq \{&\Q \in \mathcal{P}^f(X) \ | \ S \text{ is an }\mathbb{F}\text{-martingale under }\Q, \ \Q(\Omega)=1, \\
 &\E_{\Q}[\phi]=0 \ \forall \phi\in \Phi \},
\end{align*}
where $\mathcal{P}^f(X)$ denotes the finitely supported Borel probability measures on $(X, \mathcal{B}(X))$. As a general convention, in this paper we interpret the above sub- and super-scripts as restrictions on the sets of measures. When we drop some of them it is to indicate that these conditions are not imposed, e.g., $\mathcal{M}_{\Omega}(\mathbb{F})$ denotes all $\mathbb{F}$-martingale measures supported on $\Omega$. Next let
\begin{align*}
\Omega^*_{\Phi}\coloneqq\{\omega \in \Omega \ | \ \exists \Q \in \mathcal{M}_{\Omega, \Phi}^f(\mathbb{F}^\nf) \text{ s.t. }\Q(\omega)>0 \} = \bigcup_{\Q \in \mathcal{M}^f_{\Omega, \Phi}(\mathbb{F}^\nf)} \text{supp}(\Q)
\end{align*}
with the same convention regarding sub- and super-scripts as above. We also define
\begin{align*}
\mathbb{F}^M\coloneqq (\mathcal{F}_t^M)_{t \in \{0, \dots, T\}}, \ \ \text{where } \mathcal{F}^M_t= \bigcap_{\Q \in \mathcal{M}_{\Omega}(\mathbb{F}^\nf)} \mathcal{F}_t^\nf \vee \mathcal{N}^{\Q}(\mathcal{F}_T^\nf),
\end{align*}
$\mathcal{N}^{\Q}(\mathcal{F}_T^\nf) \coloneqq \{ N \subseteq A \in \mathcal{F}_T^\nf \ | \ \Q(A)=0 \}$ and $\mathcal{F}^M_t$ is the power set of $\Omega$ if $\mathcal{M}_{\Omega}(\mathbb{F}^\nf)=\emptyset$. 
\begin{Rem}
Note that $\mathbb{F}^\nf \subseteq \mathbb{F}^{\mathcal{U}}\subseteq \mathbb{F}^M$ holds. All these filtrations generate the same martingale measures on $\Omega$ calibrated to $\Phi$, which we denote by $\mathcal{M}_{\Omega, \Phi}$.
\end{Rem}
For $\P\in \Pc(X)$, thus $\Nc^\P:= \Nc^\P(\Fc^\Uc)$ denotes the collection of its null sets. Likewise, given a family $\Pfr\subset \Pc(X)$, the collection of its polar sets if given by $\Nc^\Pfr=\bigcap_{\P \in \mathfrak{P}} \mathcal{N}^\P$. We say that a property holds $\Pfr$-q.s. if it holds outside a $\Pfr$-polar set. 
\subsection{Notions of Arbitrage}\label{sec:arbitrage}
One of the most important underlying concepts in financial mathematics is the absence of arbitrage. In the literature on robust pricing and hedging many notions of arbitrage have been proposed to date. We present these here together in a unified manned and discuss their relative dependencies. To complement the picture, we establish some novel technical results. These are postponed to Section \ref{Sec. compar}.
\begin{Defn}
Fix a filtration $\mathbb{F}$, a set $\mathfrak{P}$, a set $\mathcal{S}$ of subsets of $X$ and a set $\Omega$. Recall that semistatic admissible trading strategies are given by $(h,H)\in \adm$.
\begin{enumerate}[labelsep=1cm]
\item[\textbf{1pA}($\Omega$)] A One-Point Arbitrage (see \citep{riedel2015financial}) is a strategy $(h,H) \in \adm $ such that $h \cdot \Phi+H \circ S_T \ge 0$ on $\Omega$ with strict inequality for some $\omega \in \Omega$.
\item[\textbf{OA}($\Omega$)] An Open Arbitrage (see \citep{riedel2015financial}) is a strategy $(h,H) \in \adm $ such that $h \cdot \Phi+H \circ S_T \ge 0$ on $\Omega$ with strict inequality for some open subset of  $\Omega$.
\item[\textbf{SA}($\Omega$)] A Strong Arbitrage (see \citep{acciaio2013model}) is a strategy $(h,H) \in \adm $ such that $h \cdot \Phi + H \circ S_T >0$ on $\Omega$. 
\item[\textbf{USA}($\Omega$)] A Uniformly Strong Arbitrage (see \citep{davis2007range}) is a strategy $(h,H) \in \adm$ such that $h\cdot \Phi+ H \circ S_T \ge \epsilon$ on $\Omega$ for some $\epsilon>0$.
\item[\textbf{A($\mathfrak{P}$)}] A $\mathfrak{P}$-quasi-sure Arbitrage (see \citep{bouchard2015arbitrage}) is a strategy $(h,H) \in \adm$ such that $h \cdot \Phi+ H \circ S_T \ge 0$ holds $\mathfrak{P}\text{-q.s.}$ and $\P(h \cdot \Phi+H \circ S_T >0)>0$ for some $\P \in \mathfrak{P}$. If $\mathfrak{P}=\{\P\}$ a $\mathfrak{P}$-quasi-sure Arbitrage is called a $\P$-arbitrage  and is denoted \textbf{A}($\P$).
\item[\textbf{CA}($\mathfrak{P}$)] A Classical Arbitrage in $\Pfr$ (see \citep{davis2007range}) is a family of strategies $(h^{\P}, H^{\P})_{\P\in \mathfrak{P}}$ such that, for all  $\P\in \mathfrak{P}$, $(h^{\P}, H^{\P})$ is a $\P$-arbitrage.
\item[\textbf{WA}($\mathfrak{P}$)] A Weak Arbitrage (see \citep{blanchard2019no}) is a strategy $(h,H)\in \adm$ which is a $\P$-arbitrage for some $\P\in \mathfrak{P}$.
\item[\textbf{IntA}($\mathfrak{P}$)] An Interior Arbitrage (see \citep{bayraktar2014note}) is a sequence of strategies $(h^n,H^n)\in \adm$ such that $(h^n,H^n)$ is a $\mathfrak{P}$-quasi-sure Arbitrage relative to option payoffs given by $\Phi+\text{sign}(h^n)/n$ for all $n$ large enough.
\item[\textbf{WFLVR}($\Omega$)] A Weak Free Lunch With Vanishing Risk (see \citep{cox2011robusta}, \citep{cox2014robust}) is a sequence of strategies $(h^n,H^n) \in \adm$ such that there exists a constant $c\ge 0$ and $(h,H) \in \adm$ with $h^n\cdot \Phi+ H^n \circ S_T \ge h\cdot \Phi+ H \circ S_T-c$ on $\Omega$ for all $n \in \N$ and $$\lim_{n \to \infty} (h^n\cdot \Phi+ H^n \circ S_T) >0\quad \text{on }\Omega.$$
\item[\textbf{locA($\mathfrak{P}_t(\omega)$)}] A $(t,\omega)$-local $\mathfrak{P}$-quasi-sure Arbitrage (see \citep{bartl2016exponential}) is a strategy $H \in \R^d$ such that  $H\Delta S_{t+1}(\omega) \ge 0$ $\mathfrak{P}_t(\omega)$-q.s. (where $t\in \{0,\dots,T-1\}$ and $\omega\in X$) and there exists $\P\in \mathfrak{P}_t(\omega)$ such that $\P(H\Delta S_{t+1}>0)>0$.
\item[\textbf{A($\mathcal{S}$)}] An Arbitrage de la Classe $\mathcal{S}$ (see \citep{burzoni2016universal}) is a strategy  $(h,H) \in \adm $ such that $h \cdot \Phi+H \circ S_T \ge 0$ on $\Omega$ and $\{\omega \in \Omega \ | \ h\cdot \Phi+ H \circ S_T >0 \}\supseteq \Gamma$ for some $\Gamma\in\mathcal{S}$.
\end{enumerate}
When we want to stress the role of the filtration we include it as an argument, e.g., we write, e.g., \textbf{SA}($\Omega, \mathbb{F}$).
When the filtration is not specified it is implicitly taken to be $\mathbb{F}^{\mathcal{U}}$. 
We use a prefix \textbf{N} to indicate a negation of any of the above notions, e.g., we say that ``\textbf{NA}($\Pfr$) holds" when there does not exist a $\Pfr$-quasi-sure arbitrage strategy, likewise \textbf{NUSA}($\Omega$) denotes the absence of a uniformly strong arbitrage on $\Omega$, etc.
\end{Defn}

\begin{Lem}\label{lem easy}
The following relations hold:
\begin{enumerate}
\item ${\bf USA}(\Omega)\Rightarrow {\bf SA}(\Omega) \Rightarrow {\bf OA}(\Omega)\Rightarrow {\bf 1pA}(\Omega).$\\
\item ${\bf SA}(\Omega)\Rightarrow\bf{ WFLVR}(\Omega).$\\
\item ${\bf A(\mathfrak{P})}\Rightarrow {\bf A(\P)} \text{ for some }\P\in \mathfrak{P}\Leftrightarrow\bf{WA}(\mathfrak{P}).$\\
\item $\bf{WA}(\mathfrak{P}) \Leftarrow {\bf CA}(\mathfrak{P}) \Leftrightarrow {\bf A}(\P) \text{ for all }\P\in \mathfrak{P}.$\\
\item $\mathbf{A}(\mathfrak{P})\Rightarrow \mathbf{IntA}(\mathfrak{P})$.
\item when $\Phi=0$ then\\ $\mathbf{A}(\mathfrak{P}) \Leftrightarrow \P\left(\bigcup_{t=0}^{T-1} \left\{\omega\in X_t \ | \  \mathbf{locA}(\mathfrak{P}_t(\omega))\text{ holds }\right\} \right)>0$ for some $\P\in \Pfr$.\\
\end{enumerate}
\end{Lem}

\begin{proof}
Items \textit{(1)}-\textit{(4)} are immediate. Assertion \textit{(6)} follows from \citep[Lemma 4.6, p.842]{bouchard2015arbitrage}. 
For a strategy $(h,H)\in \adm$ satisfying $$h\cdot \Phi+H \circ S_T \ge 0 \quad \mathfrak{P}\text{-q.s.}$$ 
we have, for any $\epsilon>0$, 
\begin{align*}
h\cdot (\Phi+\text{sign}(h)\epsilon)+H\circ S_T= h\cdot\Phi+H\circ S_T+|h|_1\epsilon\ge |h|_1\epsilon\ge 0,
\end{align*}
where $|h|_1=\sum_{\lambda\in \Lambda}|h_\lambda|$. 
Absence of \textbf{IntA}$(\mathfrak{P})$ implies that there exists $\epsilon>0$ such that for any strategy as above we have
\begin{align*}
h\cdot \Phi+H\circ S_T=-|h|_1\epsilon \quad \mathfrak{P}\text{-q.s.},
\end{align*}
so that $h=0$ and absence of \textbf{A}$(\mathfrak{P})$ follows so \textit{(5)} holds. 
\end{proof}

\textbf{USA}($X$) was first discussed in \citep{davis2007range}, see also \citep{cox2011robusta} and \citep{cox2014robust} for a definition of \textbf{USA}($\Omega$) and \textbf{WFLVR}($\Omega$), where $\Omega \subseteq X$. Note that if we take $(h,H)=(0,0)$ in the definition of $\mathbf{WFLVR}(\Omega)$ and replace the pathwise inequalities by their $\P$-a.s.\ counterparts for some fixed $\P\in \mathcal{P}(X)$, we recover a discrete version of the $\mathbf{NFLVR}$ condition of \citep{delbaen1994general}. \\
\textbf{SA}(($\R_+^d)^T$) was used in \citep{acciaio2013model} in the canonical setup. We refer to \citep[Theorem 3]{bfhmo} for a general FTAP connecting the notion of Strong and Uniformly Strong Arbitrage under the condition that there exists an option with a strictly convex super-linear payoff in the market. See also \citep{bartl2017duality2} for an equivalence result under marginal constraints. In Section \ref{Sec. compar} we discuss the connection between \textbf{SA}($\Omega$) and \textbf{USA}($\Omega$) without the above assumptions.\\
\textbf{A}($\mathcal{S}$) is a unifying concept since  
\textbf{1pA}($\Omega$), \textbf{OA}($\Omega$), \textbf{SA}($\Omega$), \textbf{USA}($\Omega$) and \textbf{A}($\mathfrak{P}$) can all be seen as special cases of \textbf{A}$(\mathcal{S})$, see \citep[Section 4.6]{burzoni2016universal} for a detailed discussion. It was first defined in \citep{burzoni2016universal} in a pathwise setting, see in particular the pathwise Fundamental Theorem of Asset pricing in \citep[Theorem 2 \& Section 4]{burzoni2016universal}. This extends the results obtained in \citep{riedel2015financial} who introduced \textbf{1pA}($\Omega$) and \textbf{OA}($\Omega$).  \textbf{OA}($\Omega$) is furthermore defined in the setup of \citep{dolinsky2014robust}. \\
\textbf{A}$(\mathfrak{P})$ was introduced in the quasi-sure setting of \citep{bouchard2015arbitrage}, where they prove a quasi-sure Fundamental Theorem of Asset pricing and Superhedging Theorem. From Lemma \ref{lem easy} above we see that the crucial distinction between $\textbf{CA}(\mathfrak{P})$ and $\textbf{A}(\mathfrak{P})$ is the aggregation of arbitrage strategies, which poses a fundamental technical difficulty overcome in \cite{bouchard2015arbitrage} by the specific (APS) structure of $\mathfrak{P}$. We also note that $\textbf{CA}(\mathfrak{P})$ was actually referred to as \emph{weak arbitrage} in \citep{davis2007range}.
\\
The notion of interior arbitrage $\mathbf{IntA}(\mathfrak{P})$ was introduced, and called a \emph{robust arbitrage}, by \cite{bayraktar2014note} in the context of transaction costs. Absence of $\mathbf{IntA}(\mathfrak{P})$ is equivalent to absence of \textbf{A}($\Pfr$) not only at the current prices of statically traded options $\Phi$ but also under all, sufficiently small, perturbations of their prices. This notion was also used in \cite[Assumption 3.1]{hou2015robust}. It is equivalent to saying that the prices of the options $\Phi$ are strictly inside the region of their $\mathfrak{P}$-q.s.\ no-arbitrage prices, thus avoiding the delicate issue of boundary classification. In general, $\mathbf{IntA}(\mathfrak{P})$ does not imply $\mathbf{A}(\mathfrak{P})$. To see this, take $\Phi=\{(S_T-K)^+\}$ for some $K> S_0$ and $\emptyset\neq\mathfrak{P}\subseteq\{\P\in \mathcal{P}(X) \ | \ \P(S_T\le K)=1\}$. Then there is no $\mathfrak{P}$-q.s. arbitrage, while for every $\epsilon>0$ we have $(S_T-K)^++\epsilon\ge\epsilon>0$ and thus  $\mathbf{RA}(\mathfrak{P})$ holds.

\emph{Throughout the remainder of this paper, unless otherwise stated, we take $\Lambda = \{1, \ldots, k\}$, i.e., we have a finite $\Phi$ with $k$ statically traded options.}

\subsection{Robust Fundamental Theorem of Asset Pricing}\label{sec:FTAP}
The first Fundamental Theorem of Asset Pricing characterises absence of arbitrage in terms of existence of martingale (pricing) measures. In the classical discrete-time setting, this refers to the notion of  $\mathbb{P}$-arbitrage. However, in a robust setting, there are many possible notions of arbitrage one can consider. If we adopt a strong notion of arbitrage, its absence should be equivalent to a weak statement, e.g., $\mathcal{M}_{\Omega, \Phi} \neq \emptyset$. This is often done in the pathwise literature, see \citep{bfhmo}, and leads to a robust (multi-prior) version of the familiar Dalang-Morton-Willinger theorem. 
\begin{Theorem}[Robust DMW Theorem]\label{cor:robustDMW}
Let $\mathfrak{P}$ be a set of probability measures satisfying (APS). Then there exists a universally measurable set of scenarios $\Omega$ with $\P(\Omega)=1$ for all $\P \in \mathfrak{P}$ and a filtration $\tilde{\mathbb{F}}$ with $\mathbb{F}^\nf \subseteq \tilde{\mathbb{F}} \subseteq \mathbb{F}^{\mathcal{M}}$, such that the following are equivalent:
\begin{enumerate}
\item $\mathcal{Q}_{\Pfr, \Phi} \neq \emptyset$.
\item $\P(\Omega^*_{\Phi})> 0$ for some $\P \in \mathfrak{P}$.
\item $\mathcal{M}_{\Omega,\Phi}\neq \emptyset$.
\item $\Omega^*_{\Phi} \neq \emptyset$.
\item \textbf{NSA}($\Omega, \tilde{\mathbb{F}}$) holds.
\end{enumerate}
Conversely, for an analytic set $\Omega$ there exists a set $\mathfrak{P}$ satisfying (APS) such that for all $\omega\in \Omega$ there exists $\P \in \mathfrak{P}$ with $\P(\{\omega\})>0$ and such that (1)-(5) are equivalent.
\end{Theorem}
The above result follows from Theorem \ref{Thm. S} below by setting $\mathcal{S}=\{\Omega\}$. To see its opposite twin we should adopt a weak notion of arbitrage, its absence thus being equivalent to a strong statement, e.g., for all $\P \in \mathfrak{P}$ there exists $\Q \in \mathcal{Q}_{\Pfr,\Phi}$ such that $\P\ll \Q$. This route is most often taken in the quasi-sure literature, see \citep{bouchard2015arbitrage}, and leads to the following version of the robust FTAP. 
\begin{Theorem} \label{Thm. bn vs bfhmo}
Let $\mathfrak{P}$ be a set of probability measures satisfying (APS). Then there exists an analytic set of scenarios $\Omega$ with $\P(\Omega)=1$ for all $\P \in \mathfrak{P}$, such that the following are equivalent:
\begin{enumerate}
\item $\textbf{N1pA}(\Omega^*_{\Phi})$ holds and $\Omega=\Omega^{*}_{\Phi}$ $\Pfr$-q.s.
\item For all $\P \in \mathfrak{P}$ there exists $\Q \in \mathcal{Q}_{\Pfr, \Phi}$ such that $\P\ll \Q$. 
\item \textbf{NA}($\mathfrak{P},\mathbb{F}^{\mathcal{U}})$ holds.
\end{enumerate}
Conversely, if $\Omega$ is an analytic set, then there exists a set $\mathfrak{P}$ of probability measures satisfying (APS) such that for all $\omega\in \Omega$ there exists $\P \in \mathfrak{P}$ with $\P(\{\omega\})>0$ and such that the following are equivalent: 
\begin{enumerate}
\item $\textbf{N1pA}(\Omega)$ holds and $\Omega= \Omega_{\Phi}^*$.
\item For all $\P \in \mathfrak{P}$ there exists $\Q \in \mathcal{Q}_{\Pfr, \Phi}$ such that $\P \ll \Q$.
\item \textbf{NA}$(\mathfrak{P},\mathbb{F}^{\mathcal{U}})$ holds.
\end{enumerate}
\end{Theorem}
Our proof of this theorem, given in Section \ref{subseq. ftap}, does not rely on the proof of (3) $\Rightarrow$ (2) given in \citep{bouchard2015arbitrage}. Instead we give pathwise arguments. In particular, given $\P \in \mathfrak{P}$ such that $\P(\Omega \setminus \Omega^*_{\Phi})>0$ we explicitly construct a quasi-sure Arbitrage strategy using the Universal Arbitrage Aggregator of \citep{bfhmo}. This strengthens the results of Theorem \ref{Thm. S} below. Indeed, using the fact that $\mathfrak{P}$ satisfies (APS), it is possible to select $\Q \in \mathcal{Q}_{\Pfr, \Phi}$ for each $\P \in \mathfrak{P}$ such that $\P \ll \Q$. Necessarily the support of each $\P$ is then concentrated on $\Omega_{\Phi}^*$.

Finally, we give our main abstract result, which establishes a pathwise and probabilistic characterisation of the absence of Arbitrage de la Classe $\mathcal{S}$. Its proof is presented in Section \ref{subseq. ftap}. As noted above, Arbitrage de la Classe $\mathcal{S}$ allows to consider many notions of arbitrage at once. 
Accordingly, the main result below implies  Theorem \ref{cor:robustDMW} and can be strengthened to imply Theorem \ref{Thm. bn vs bfhmo} as will be seen in Section \ref{sec:proofs}.
\begin{Theorem} \label{Thm. S}
Assume that $\mathfrak{P}$ satisfies (APS) and $\mathcal{S} \subseteq \mathcal{B}(X)$ is such that 
\begin{align} \label{eq. app}
\exists \{C_n\}_{n \in \N} \subseteq \mathcal{S}\text{ s.t. } \forall C \in \mathcal{S} \ \exists \{n_k\}_{k \in \N}\subseteq \N \text{ with }  \mathds{1}_{C_{n_k}} \uparrow \mathds{1}_{C} \ (k \to \infty).
\end{align}
Then there exists a co-analytic set of scenarios $\Omega$ such that $\P(\Omega)=1$ for all $\P \in \mathfrak{P}$ and  a filtration $\tilde{\mathbb{F}}$ with $\mathbb{F}^\nf \subseteq \tilde{\mathbb{F}} \subseteq \mathbb{F}^{\mathcal{M}}$, such that the following are equivalent:
\begin{enumerate}
\item For all $ C \in \mathcal{S}$ with $C \subseteq \Omega$ there exists $ \Q \in \mathcal{Q}_{\Pfr, \Phi}$ such that $\Q(C)>0$.
\item For all $C \in \mathcal{S}$ with $C \subseteq \Omega$ there exists $\P \in \mathfrak{P}$ with $\P(\Omega_{\Phi}^*\cap C)>0$.
\item For all $C \in \mathcal{S}$ with $C \subseteq \Omega$ there exists $\Q \in \mathcal{M}_{\Omega, \Phi}$ such that $\Q(C)>0$.
\item $\{C \in \mathcal{B}(X) \ | \ C \subseteq \Omega \setminus \Omega_{\Phi}^*\} \cap \mathcal{S}= \emptyset$.
\item There is no Arbitrage de la Classe $\mathcal{S}$ in $\Ac_{\Phi}(\tilde{\mathbb{F}})$ on $\Omega$.
\end{enumerate}
Conversely, for an analytic set $\Omega$ there exists a set $\mathfrak{P}$ satisfying (APS),  such that for all $\omega\in \Omega$ there exists $\P \in \mathfrak{P}$ with $\P(\{\omega\})>0$ and such that (1)-(5) are equivalent.
\end{Theorem}

\begin{Rem}
Condition \eqref{eq. app} was first stated in \citep[Cor. 4.30 and the discussion thereafter]{burzoni2016universal}. 
It turns out that for the proof of Theorem 2.7 a weaker condition is sufficient: we only need the properties
\begin{align} \label{eq: exact1}
 \bigcup \{ C \in \mathcal{S} \ | \ C \cap (\Omega^{\mathfrak{P}})_{\Phi}^*\in \mathcal{N}^{\mathfrak{P}} \} \in \mathcal{B}(X)
\end{align}
and 
\begin{align} \label{eq: exact2}
\bigcup \{ C \in \mathcal{S} \ | \ C \cap (\Omega^{\mathfrak{P}})_{\Phi}^*\in \mathcal{N}^{\mathfrak{P}} \} \cap(\Omega^{\mathfrak{P}})_{\Phi}^* \in \mathcal{N}^{\mathfrak{P}}
\end{align}
to hold, where we refer to Section \ref{subseq. ftap} for a formal definition of $\Omega^{\mathfrak{P}}$. Conditions \eqref{eq: exact1} and \eqref{eq: exact2} are compatibility conditions on $\Omega$, $\mathcal{S}, \Phi$ and $\mathfrak{P}$. Indeed, they assert that the (likely uncountable) union of 
``inefficient" subsets of $\Omega^{\Pfr}\setminus (\Omega^{\mathfrak{P}})_{\Phi}^*$ (modulo $\mathfrak{P}$-polar sets), stays an ``inefficient" subset (modulo $\mathfrak{P}$-polar sets). If this condition is not satisfied  for some arbitrary $\Pfr$ and $\Sc$, then there is no reason why a set $\Omega$ for which \textit{(2)} holds should exist. Take for example a collection $\mathfrak{P}$ having densities and $\mathcal{S}$ a set of singletons in $X$. Then $\P(C)=0$ for any $\P\in \Pfr$ and any $C\in \mathcal{S}$ so the only $\Omega$ which could satisfy \textit{(2)} is the empty set. 
We note that when $\mathcal{S}= \{ C \subseteq X \ | \ C \text{ open} \}$ then  \eqref{eq: exact1} is always satisfied and \eqref{eq: exact2} is satisfied as soon as $X$ is separable. However, in general, conditions \eqref{eq: exact1} and \eqref{eq: exact2} may be hard to verify, which is  why we provide \eqref{eq. app} as an easier sufficient condition. Lastly, we remark that it is not straightforward to show that $\mathcal{S}= \{ C \ | \ \P(C)>0 \text{ for some }\P \in \mathfrak{P}\}$ corresponding to $\mathbf{NA}(\mathfrak{P})$ satisfies \eqref{eq: exact2}, which is why we give a direct proof of Theorem \ref{Thm. bn vs bfhmo} in  Section \ref{sec:proofs}.
\end{Rem}

We note that the set $\Omega$ can in general not be assumed to be analytic. The implications $\textit{(1)} \Rightarrow \textit{(2)} \Rightarrow \textit{(3)} \Rightarrow \textit{(4)} \Rightarrow \textit{(5)}$ follow directly from the definitions. Apart from measurability considerations regarding $\Omega$, the equivalence of \textit{(3)}, \textit{(4)} and \textit{(5)} essentially follows from \citep{burzoni2016universal}.  Furthermore, given an analytic set $\Omega$, we will simply define $\mathfrak{P}$ as all the finitely supported probability measures on $\Omega$. The analyticity of $\Omega$ then implies (APS) of $\mathfrak{P}$. We then also have $\Qc_{\Pfr, \Phi}=\mathcal{M}^f_{\Omega, \Phi}$ and equivalence of \textit{(1)} and \textit{(3)}-\textit{(5)} follows from \citep{bfhmo}. 
In this context, the essential connection we make is the combination of pathwise and quasi-sure criteria as stated in \textit{(2)}: for every $C \in \mathcal{S}$, the pathwise efficient subset $\Omega_{\Phi}^*\cap C$ is required to be ``seen" by at least one measure $\P$ in the set $\mathfrak{P}$. 

For a given $\Pfr$, the set $\Omega$ in Theorem \ref{Thm. S} can be explicitly constructed as the concatenation of the quasi-sure supports of $\mathfrak{P}_t \circ \Delta(S_{t+1})^{-1}$. The main difficulty of the proof is then to show the implication $\textit{(5)}\Rightarrow \textit{(1)}$, where one needs to establish existence of martingale measures $\Q \in \mathcal{Q}_{\Pfr,\Phi}$, which are compatible with $\Omega$ and $\mathcal{S}$ in the sense of \textit{(1)}. 
This, modulo measurable selection arguments, is achieved by finding an element $\P \in \mathfrak{P}_t(\omega)$ such that zero is in the relative interior of the support of $\P \circ \Delta(S_{t+1})^{-1}$. Indeed, let us explain the main idea of the proof of $\textit{(5)}\Rightarrow \textit{(1)}$
based on the following example: assume $T=1$, $d=3$, $\Phi=0,\mathcal{S}=\{\Omega\}$ and the set $\Omega^*$ is given by the grey polyhydron in Figure \ref{fig::ftap1}. Assume that the support of $\P \circ \Delta(S_{1})^{-1}$ for a given measure $\P \in \mathfrak{P}_0$ is given by the blue dot (see Figure \ref{fig::ftap1}). Then as $0 \in \text{ri}(\Omega^*)$, we can find three additional points in $\Omega^*$, such that zero is in the relative interior of the convex hull of the four points. By definition of $\Omega$, the three additional points are in the support of some measures in $\mathfrak{P}_0$, which we call $\P_1, \P_2, \P_3$ in $\mathfrak{P}_0$. As $\mathfrak{P}_0$ is convex, it follows that $$\tilde{\P}\coloneqq \frac{\P_1+\P_2+\P_3+\P}{4}$$ is an element of $\mathfrak{P}_0$ as well, as visualised in Figure \ref{fig::ftap2}. Since zero is in the relative interior of the support of $\tilde{\P}$, one can now use results from \citep{rokhlin2008proof} to find a martingale measure $\Q \sim \tilde{\P}$, in particular $\P \ll \Q$.


\tdplotsetmaincoords{60}{138}

\pgfmathsetmacro{\rvec}{.8} \pgfmathsetmacro{\thetavec}{30}
\pgfmathsetmacro{\phivec}{60}


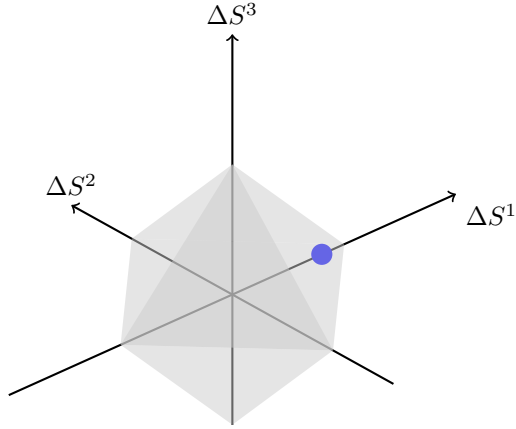
\begin{figure}[h!] 
\centering

\definecolor{cccccc}{rgb}{0.8,0.8,0.8}

\begin{tikzpicture}[scale=4,tdplot_main_coords]

\draw[thick,->] (1,0,0) -- (-1,0,0) node[anchor=north
west]{$\Delta S^1$}; \draw[thick,->] (0,0.8,0) -- (0,-0.8,0)
node[anchor=south]{$\Delta S^2$}; \draw[thick,->] (0,0,-0.5) --
(0,0,1) node[anchor=south]{$\Delta S^3$};

\fill [color=blue] (-0.4,0,0) circle (1pt); 

\fill[color=cccccc,fill=cccccc,fill
opacity=0.5] (0.5,0,0) -- (0,0.5,0)-- (-0.5,0,0) -- (0,-0.5,0) -- 
cycle;
\fill[color=cccccc,fill=cccccc,fill
opacity=0.5] (0.5,0,0) -- (0,0.5,0)-- (0,0,0.5) -- 
cycle;
\fill[color=cccccc,fill=cccccc,fill
opacity=0.5] (-0.5,0,0) -- (0,-0.5,0) -- (0,0,0.5) -- 
cycle;
\fill[color=cccccc,fill=cccccc,fill
opacity=0.5] (0.5,0,0) -- (0,0.5,0) -- (0,0,-0.5) -- 
cycle;

\end{tikzpicture}
\caption{Construction of a martingale measure $\Q \gg \P$ for $d=3$: the set $\Omega$ (grey) and $\text{supp}(\P \circ (\Delta S_1)^{-1})$ (blue)\label{fig::ftap1}}
\end{figure}

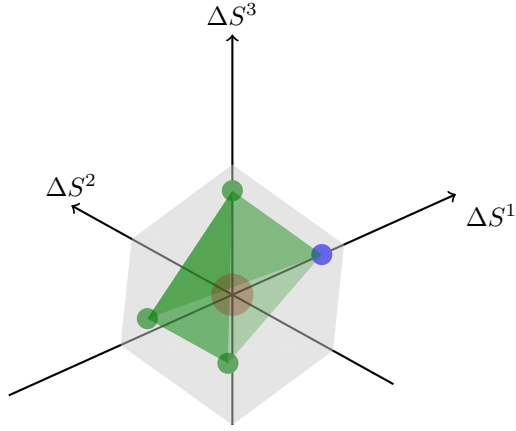
\begin{figure}[h!] 
\centering

\definecolor{cccccc}{rgb}{0.8,0.8,0.8}
\definecolor{darkred}{rgb}{0.64,0,0}

\begin{tikzpicture}[scale=4,tdplot_main_coords]

\draw[thick,->] (1,0,0) -- (-1,0,0) node[anchor=north
west]{$\Delta S^1$}; \draw[thick,->] (0,0.8,0) -- (0,-0.8,0)
node[anchor=south]{$\Delta S^2$}; \draw[thick,->] (0,0,-0.5) --
(0,0,1) node[anchor=south]{$\Delta S^3$};

\fill [color=ao] (0.2,0.2,-0.1) circle (1pt); 
\fill [color=blue] (-0.4,0,0) circle (1pt); 
\fill [color=ao] (0.2,-0.2,-0.1) circle (1pt); 
\fill [color=ao] (0,0,0.4) circle (1pt);

\fill[color=cccccc,fill=cccccc,fill
opacity=0.5] (0.5,0,0) -- (0,0.5,0)-- (-0.5,0,0) -- (0,-0.5,0) -- 
cycle;
\fill[color=cccccc,fill=cccccc,fill
opacity=0.5] (-0.5,0,0) -- (0,-0.5,0) -- (0,0,0.5) -- 
cycle;
\fill[color=cccccc,fill=cccccc,fill
opacity=0.5] (0.5,0,0) -- (0,0.5,0) -- (0,0,-0.5) -- 
cycle;

\fill[color=ao,fill=ao,fill
opacity=0.5] (0.2,0.2,-0.1) -- (0.2,-0.2,-0.1) -- (0,0,0.4) -- 
cycle;
\fill[color=ao,fill=ao,fill
opacity=0.25] (0.2,0.2,-0.1) -- (-0.4,-0,0) -- (0,0,0.4) -- 
cycle;
\fill[color=ao,fill=ao,fill
opacity=0.25] (0.2,-0.2,-0.1) -- (-0.4,-0,0) -- (0,0,0.4) -- 
cycle;

\fill [color=darkred,opacity=0.25] (0,0,0) circle (2pt);

\end{tikzpicture}
\caption{Construction of a martingale measure $\Q \gg \P$ for $d=3$: finding a measure $\tilde{\P}$ such that NA$(\tilde{\P})$ and $\tilde{\P}\gg \P$ holds.
\label{fig::ftap2}}
\end{figure}
Note that this argument fundamentally relies on the convexity of $\mathfrak{P}_t$. The analytic product structure assumption then grants suitable measurability for the concatenation procedure in the multiperiod case.


\subsection{Robust Superhedging Theorem}\label{subseq superh}

In this section we focus on the key result which characterises superhedging prices: the pricing-hedging duality, or the Superhedging Theorem. As before, we 
compare pathwise and quasi-sure superhedging approaches as extensions of the classical model-specific result, see \cite[Chapter 5, Theorem 5.30]{follmer2011stochastic}. 

For a set $\Omega \subseteq X$ we denote the pathwise superhedging price on $\Omega$ by
\begin{align*}
\pi_{\Omega}(g):=\inf \{ x \in \R \ | \ \exists (h,H) \in \admU  \text{ s.t. } x+ h \cdot \Phi+(H \circ S_T) \ge g \ \text{ on } \Omega \}
\end{align*} 
and denote the $\mathfrak{P}$-q.s. superhedging price by
\begin{align*}
\pi^{\mathfrak{P}}(g):= \inf \{ x \in \R \ | \ \exists (h,H) \in \admU \text{ s.t. } x + h \cdot \Phi + (H \circ S_T) \ge g \ \mathfrak{P}\text{-q.s.}\}.
\end{align*}
Take an analytic set $\Omega$ such that for all $\P \in \mathfrak{P}$ we have $\P(\Omega^*_{\Phi})=1$. Using the Superhedging Theorems of \citep{bouchard2015arbitrage} and \citep{bfhmo} it is immediate that the following relationships hold for all upper semianalytic $g$: 
\begin{align*}
\sup_{\Q \in \mathcal{M}_{\Omega,\Phi}} \E_{\Q}[g]=\pi_{\Omega_{\Phi}^*}(g) &\ge \pi^{\mathfrak{P}}(g)=\sup_{\Q \in \mathcal{Q}_{\Pfr,\Phi}} \E_{\Q}[g].
\end{align*}
The above inequality is strict in general. An easy way to see this is to take $d=T=S_0=1$, $\Phi=0$, $g(S_1)=\mathds{1}_{\{S_1=0\}}$ and $\mathfrak{P}=\{\frac{1}{2}\lambda|_{[0,2]}\}$, where $\lambda|_{[0,2]}$ denotes the Lebesgue measure on $[0,2]$. Then $\Omega=\Omega^*=[0,2]$ and the pathwise superhedging price is equal to $1/2$, while the quasi-sure superhedging price is equal to zero. In fact, to link the super-hedging and pathwise formulations, we have to choose a specific set $\Omega^{\mathfrak{P}}_g$ which depends not only on $\Pfr$ but also on $g$. We determine this set $\Omega^{\mathfrak{P}}_g$ by reducing to superhedging under a fixed measure $\P^g$ as stated in the following theorem: 

\begin{Theorem}\label{Thm super}
Let $\mathfrak{P}$ be a set of probability measures satisfying (APS). Assume \textbf{NA}$(\mathfrak{P})$ holds and let $g: X \to \R$ be upper semianalytic.
Then there exists a measure $\P^g= \P^g_0 \otimes \cdots \otimes \P^g_{T-1}$ and an $\mathcal{F}^{\mathcal{U}}$-measurable set $\Omega_g^{\mathfrak{P}}$ with $\P(\Omega^{\mathfrak{P}}_g)=1$ for all $\P \in \mathfrak{P}$, such that 
\begin{align*}
\sup_{\Q \in \mathcal{M}_{\Omega^{\mathfrak{P}}_g,\Phi}} \E_{\Q}[g]=\pi_{(\Omega_g^{\mathfrak{P}})_{\Phi}^*}(g)=\pi^{\P^g}(g)= \pi^{\mathfrak{P}}(g)=
\sup_{\Q \in \mathcal{Q}_{\Pfr, \Phi}}\E_{\Q}[g].\nonumber
\end{align*}
Conversely, let $\Omega$ be an analytic subset of $X$ with $\Omega^*_{\Phi} \neq \emptyset$ and let $g: X \to \R$ be upper semianalytic. For any set $\mathfrak{P}\subseteq \mathcal{P}(X)$, which satisfies (APS) and $\Nc^\Pfr=\Nc^{\mathcal{M}_{\Omega, \Phi}^f}$, we have
\begin{align*}
\sup_{\Q \in \mathcal{M}^f_{\Omega, \Phi}} \E_{\Q}[g]
&= \pi_{\Omega_{\Phi}^*}(g) = \pi^{\mathfrak{P}}(g)=\sup_{\Q \in \mathcal{Q}_{\Pfr, \Phi}}\E_{\Q}[g].
\end{align*}
In both cases, the value, if finite, is attained by a superhedging strategy $(h,H) \in \admU$.
\end{Theorem}
The proof of this result is postponed to Section \ref{Sec proofsuper}.
In particular, Theorem \ref{Thm super} lets us interpret robust superhedging prices $ \pi^{\mathfrak{P}}(g)$ as classical superreplication prices $\pi^{\P^g}(g)$ under an ``extremal" measure $\P^g$. Determining such measures $\P^g$ is not straightforward in general. 
In a one-period case and for a continuous $g$, we can use the arguments in the proof of Lemma \ref{lem. p-superhedging} to see that any measure $\P$ which attains the one-step quasisure support $\{\P\circ (\Delta S_{T}(\omega, \cdot))^{-1}\ | \ \P\in \mathfrak{P}_{T-1}(\omega)\}$ can be chosen. To extend this result to the multiperiod-case, certain continuity properties of the maps $\omega \mapsto \mathfrak{P}_t(\omega)$ have to be guaranteed: we refer to \cite[Prop. 3.7]{cow} for a sufficient condition.

\section{Complementary results on superhedging and arbitrage}\label{sec:additional}

\subsection{Extension of Pathwise Superhedging from $\Omega^*$ to $\Omega$}\label{sec. ext}
The preceding results show that quasi-sure and pathwise superhedging are essentially equivalent. As $\mathfrak{P}$-q.s.~superhedging strategies might be difficult to compute and implement in practice, it might be preferable to work on a prediction set $\Omega$ using pathwise arguments. Given that determining $\Omega^*$ is computationally expensive as well, the quantity of interest is then the superhedging price on $\Omega$ and not on $\Omega^*$ seen in the duality results in Section \ref{subseq superh}. Thus, we would like to find sufficient conditions under which the superhedging strategy associated with $\pi_{\Omega^{*}_{\Phi}}(g)$ can be extended to $\Omega$ without any additional cost.  The intuition is that as $\Omega \setminus \Omega^*$ describes non-efficient beliefs, we should be able to superhedge $g$ on this set implementing an arbitrage strategy. It turns out that this intuition is not true in general. Indeed, we run into problems regarding measurability of these arbitrage strategies, which means that this procedure only works in special cases. 

To simplify the analysis, throughout this section only, we assume that $\Phi=0$ and $\omega \mapsto S_t(\omega)$ is continuous. The latter is satisfied, e.g., when $\omega \mapsto S_t(\omega)$ is the coordinate mapping, i.e., $S_t(\omega)=\omega_t$. In order to give some intuition and to identify necessary conditions for the sets $\Omega$, $\Omega^*$ and the function $g$ we first give two counterexamples:

\begin{Ex}
Let $d=1$, $T=1$ and $(\Omega, \F)=(\R_+\setminus\{0\}, \mathcal{B}(\R_+\setminus\{0\}))$. We set $S_0=2$ and $S_1(\omega)= 2+\omega$. Then $\Omega^*=\emptyset$ and trivially 
\begin{align*}
\pi _{\Omega^*}(1)=&\inf \{ x \in \R \ | \ \exists H \in \R^d \text{ such that }x+H \circ S_T \ge 1 \text{ on } \Omega^*\} = -\infty, \\
\pi _{\Omega}(1)= &\inf \{ x \in \R \ | \ \exists H \in \R^d \text{ such that }x+H \circ S_T \ge 1 \text{ on } \Omega\} = 1.
\end{align*}
Thus we have to assume that $\Omega^*\cap \Sigma_t^{\omega}\neq \emptyset$ in the remainder of this section. We also note that here \textbf{SA}($\Omega$) holds whilst \textbf{USA}($\Omega$) does not, see Section \ref{Sec. compar}.
\end{Ex}

\begin{Ex}
\label{Ex. joh2}
Let $d=2$, $T=1$ and $\Omega=\left((2, \infty) \times[0, \infty) \right) \cup \left(\{2 \} \times [0,7]\right)$ and $\F=\mathcal{B}(\Omega)$. We set $S_0=(2,2)$ and $S_1(\omega)=\omega$. In particular, 
$\Delta S_1 (\Omega)$ is not a closed set. 
Note that $\Omega^*= \{2\} \times [0,7]$. Now, introduce the claim
\begin{align*}
g(S^1, S^2)= \Delta S_1^2 \mathds{1}_{\{\Delta S_1^2 \le 5\}}+ 5(\Delta S_1^2 -4) \mathds{1}_{\{\Delta S^2_1 >5\}}.
\end{align*}
It is easy to see that $\pi_{\Omega^*}(g)=0$ and any trading strategy $H_1=(H^1_1, 1)$ with $H^1_1\in \R$ is a superhedging strategy. We now claim that we cannot extend superhedging to $\Omega\setminus \Omega^*$ with initial capital zero. For this we show that even for initial capital one, there exist no superreplication strategies on $\Omega$. Indeed for this we would need
\begin{align*}
1+H^1_1 \Delta S^1_1+H_1^2\Delta S^2_1 &\ge 5 (\Delta S^2_1-4)\quad \text{on } \Delta S_1^1 >0,\ \Delta S_1^2>5,
\end{align*}
 which is equivalent to 
\begin{align*}
H_1^1 \ge \frac{(5-H_1^2)\Delta S_1^2-21}{\Delta S_1^1}.
\end{align*}
As $H_1^2\in [4/5, 3/2]$, this means $H^1_1\to \infty$ if $\Delta S_1^1$ is arbitrarily close to 0 and $\Delta S_1^2$ is sufficiently large. Even if we look at $\Omega=[2+\epsilon, \infty) \times[0, \infty) \cup \{2 \} \times [0,7]$ for some positive $\epsilon$, then taking $\Delta S_1^2$ arbitrarily large still leads to non-existence of superhedging strategies.
In conclusion we will only consider compact sets $\Omega\cap \Sigma_t^{\omega}$ in the rest of this section, on which ``arbitrage strategies are effective for superhedging" in a sense defined below. Furthermore, modifying the function $g$ on $\Omega \setminus \Omega^*$ in the above example, we can easily construct situations, in which $\pi_{\Omega}(g)\neq\pi_{\Omega^*}(g)$ for discontinuous payoffs $g$. In conclusion we will also assume that $g$ is continuous in this section.

We can also modify this example so that the $\Delta S_1(\Omega)$ is closed and there is no attainment of superreplication strategies for $\pi_{\Omega}(g)$. We stress that this is a fundamental difference to the case $\Omega= \Omega^*$, where attainment is always given (see Theorem \ref{Thm super}). Namely, take $\Omega=\{ (x,y) \in \R^2 \ | \ x \in [2, \infty), \ 0 \le y \le 5+\sqrt{x} \}$ with the other elements unchanged. Note that $\Omega^*$ did not change. Repeating the arguments above and looking at $\Delta S^1_1= 1/n$ and $\Delta S^2_1=5+1/ \sqrt{n}$ we find
\begin{align*}
H_1^1 \ge n\left(20+\frac{5}{\sqrt{n}}-20-\frac{1}{\sqrt{n}}\right)=n \frac{4}{\sqrt{n}} \to \infty
\end{align*}
for $n \to \infty$.
\end{Ex}

As we have seen in the examples above, in general it is necessary to assume that $\Omega^*\cap \Sigma_t^{\omega} \neq \emptyset$, $\omega \mapsto g(\omega)$ is continuous and that $\Omega$ is compact as well as ``well suited for superhedging by arbitrage strategies". We first address the second point and show continuity of the one-step superhedging prices $\omega \mapsto \pi_{t, \Omega^*}(g)(\omega)$, which are defined via a dynamic programming approach:

\begin{Defn}
For a Borel-measurable $g:X\to \R$ we define the one-step superhedging prices
\begin{align*}
\pi_{T,\Omega^*}(g)(\omega)&\coloneqq g(\omega),\\
\pi_{t,\Omega^*}(g)(\omega)&\coloneqq \inf \{x\in \R \ | \ \exists H\in \R^d \text{ s.t. }x+H\Delta S_{t+1}(v)\ge \pi_{t+1,\Omega^*}(g)(v)\ \forall v\in\Sigma_t^{\omega}\cap \Omega^* \},
\end{align*}
where $0\le t\le T-1$.
\end{Defn}

A sufficient condition for continuity of $\omega \mapsto \pi_{t, \Omega^*}(g)(\omega)$ was identified in \citep{cow} and relies on the following assumption:

\begin{Ass}\label{Ass. joh1}
The sets $\Sigma_t^{\omega} \cap \Omega^* \neq \emptyset $ and the sets $\Sigma_t^{\omega}\cap\Omega$ are compact  for all $\omega \in \Omega$ and all $0\le t\le T-1$. Furthermore, for all $0 \le t\le T-1$, the correspondence 
$\omega \twoheadrightarrow S_{t+1}(\Sigma_t^{\omega} \cap \Omega^*)$ is uniformly continuous from $(\Omega, d_t^S)$ to the subsets of $\R^d$ endowed with the Hausdorff distance, and where $d_t^S(\omega,\tilde{\omega}):=\max_{s=0, \dots, t} | S_s(\omega)- S_s(\tilde{\omega})|$.
\end{Ass}

We refer to \citep[Section 3]{cow} for a discussion and examples of sets $\Omega$ satisfying Assumption \ref{Ass. joh1}. The following lemma now follows from a direct application of \citep[Proposition 3.5]{cow}:

\begin{Lem}
Let $\omega \mapsto g(\omega)$ be continuous. Under Assumption \ref{Ass. joh1} the one-step superhedging prices $\omega \mapsto \pi_{t, \Omega^*}(g)(\omega)$ are continuous for all $0\le t \le T-1$.
\end{Lem}

Secondly, Example \ref{Ex. joh2} also shows, that it is important to identify the subset of $\Sigma_t^{\omega}\cap \Omega$, on which ``arbitrage strategies are ineffective for superhedging" in the following sense:

\begin{Defn}\label{def:A}
We denote by $\text{proj}_{\Delta S_{t+1}(\Sigma_t^{\omega} \cap \Omega^*)}(\Delta S_{t+1}(v))$ the orthogonal projection of $\Delta S_{t+1}(v)$ onto the linear subspace spanned by $\Delta S_{t+1}(\Sigma_t^{\omega} \cap \Omega^*)$  and define the set $A_t^{\omega}$ as the collection of all $v \in \Sigma_t^{\omega} \cap \Omega$, for which $\text{proj}_{\Delta S_{t+1}(\Sigma_t^{\omega} \cap \Omega^*)}(\Delta S_{t+1}(v))$ is not an element of $\Delta S_{t+1}(\Sigma_{t}^{\omega}\cap \Omega^*)$.
\end{Defn}

For an illustration of the set $A_t^{\omega}$ see Figure \ref{figure:At}. 
We now state an assumption ensuring compatibility of $A_t^{\omega}$ and $\Sigma_t^{\omega}\cap \Omega^*$:

\begin{Ass}\label{Ass. joh2}
For each level set the following is true: if a sequence of points $(v_n) \subseteq A_t^{\omega}$ converges to a point $v\in \text{span}(\Delta S_{t+1}(\Sigma_t^{\omega} \cap \Omega^*))$, then necessarily $v \in \Delta S_{t+1}(\Sigma_t^{\omega} \cap \Omega^*)$.
\end{Ass}

\begin{figure}[h!]
\centering
\begin{tikzpicture}[line cap=round,line join=round,x=1cm,y=1cm]


\draw [color=cqcqcq,dash pattern=on 4pt off 4pt,
xstep=2.0cm,ystep=2.0cm] (0,-3) grid (7,7); 
\draw[->,color=black] (-1,0) -- (7.1,0); 
\foreach \x in {2,4,6}
\draw[shift={(\x,0)},color=black] (0pt,2pt) -- (0pt,-2pt)
node[below] {\footnotesize $\x$}; 
\draw[->,color=black] (0,-3) --
(0,7.1); 
\foreach \y in {-2,2,4,6,} \draw[shift={(0,\y)},color=black]
(2pt,0pt) -- (-2pt,0pt) node[left] {\footnotesize $\y$};
\draw[color=black] (0pt,-10pt) node[right] {\footnotesize $0$};
\clip(-1.2,-5) rectangle (8,8); \draw (-1.35,7) node[anchor=north
west] {$ \Delta S^2_{1} $}; \draw (7.1,-0.16) node[anchor=north west]
{$ \Delta S^1_{1} $};
\draw[-, color=amber, line width=2pt, color=babyblue ] (0,-3)--(0,7);
\draw[dashed,color=ao, line width=2pt] (0,-2) -- (0,5);
\draw[dashed, color=qqqqff, line width=2pt] (0,2) -- (5,2);
\fill[color=qqqqff] (0,2) circle (3.0pt);
\begin{scriptsize}
\fill [color=qqqqff] (5,2) circle (3pt); 
\fill[color=cccccc,fill=cccccc,fill
opacity=0.5] (0,7) -- (7,7)-- (7,5) -- (0,5) --
cycle;
\draw[->, color=amber, line width=1.2pt ] (0,0)--(1,0);
\draw (5.1,2) node[anchor=north west, color=qqqqff] {$ \Delta S_{1}(v) $};
\draw (0,2) node[anchor=south west, color=qqqqff] {$\text{proj}_{\Delta S_{1}(\Omega^*)}( \Delta S_{1}(v)) $};
\draw (0.5,0) node[anchor=north west, color=amber] {$ \xi_{1, \Omega} $};
\end{scriptsize}
\draw (0,-3.2) node[anchor=north west] {Legend:};
\draw[-, color=cccccc, line width=2pt] (1.7,-3.55) -- (1.7,-3.35);
\draw (2,-3.2) node[anchor=north west] { $A_1^{\omega}$};
\draw (2,-3.7) node[anchor=north west] { $\Delta S_{1}(\Omega^*)$};
\draw[-, color=ao, line width=2pt] (1.7,-3.85) -- (1.7,-4.05);
\draw (2,-4.2) node[anchor=north west] { $\text{Lin}(\Delta S_{1}(\Omega^*))$};
\draw[-, color=babyblue, line width=2pt] (1.7,-4.35) -- (1.7,-4.55);

\end{tikzpicture}
\caption[ Example \ref{Ex. joh2} ]{ Example \ref{Ex. joh2} with notation from Definition \ref{def:A}.\label{figure:At}}
\end{figure}
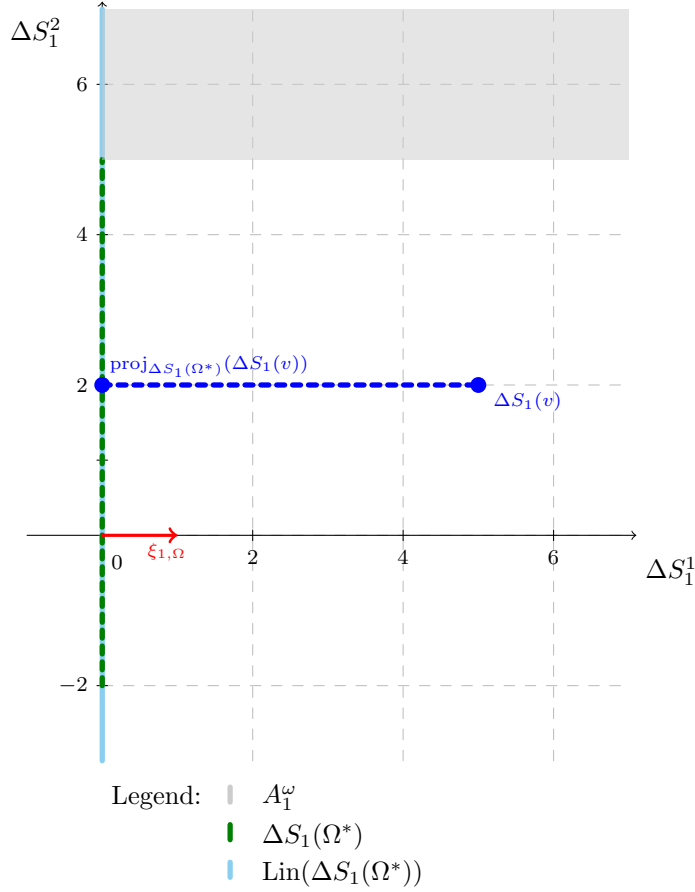

Alas, it turns out that while Assumptions \ref{Ass. joh1} and \ref{Ass. joh2} are sufficient to establish the equality $\pi_{\Omega}(g)=\pi_{\Omega^*}(g)$ for $d=2$, it is not so for $d>2$. It is linked with the notion of standard separators introduced in \citep{bfhmo}, which are measurable selectors of pointwise arbitrage strategies. We refer the reader to \citep[Proof of Lemma 1]{bfhmo} and the discussion therein for a detailed definition. Here we formulate an example, in which the existence of two standard separators together with the measurability constraint on $H_1$ implies $\pi_{\Omega}(g)>\pi_{\Omega^*}(g)$:

\begin{Ex}
\label{Ex. joh1}
Let $d=3$, $T=1$ and $(\Omega, \F)=(\R_+, \mathcal{B}(\R_+))$. We set $(S_0^1,S_0^2, S_0^3)=(2,2,2)$ and
\begin{align*}
S_1^1(\omega)= \left\{
			\begin{array}{ll}
					2 &\text{ if }\omega \in \R_+ \setminus \Q, \\
					2.5-\omega &\text{ if }\omega \in \Q \cap [1/2, \infty), \\
					4 &\text{ if }\omega \in \Q \cap [0, 1/2),
			\end{array}
				\right. \\
S_1^2(\omega)= \left\{
			\begin{array}{ll}
					\omega &\text{ if }\omega \in \R_+ \setminus \Q, \\
					0 &\text{ if }\omega \in \Q \cap [1/2, \infty), \\
					2 &\text{ if }\omega \in \Q \cap [0, 1/2),
			\end{array}
				\right. 
S_1^3(\omega)= \left\{
			\begin{array}{ll}
					2 &\text{ if }\omega \in \R_+ \setminus \Q, \\
					2 &\text{ if }\omega \in \Q \cap [1/2, \infty), \\
					2+\omega &\text{ if }\omega \in \Q \cap [0, 1/2).
			\end{array}
				\right. 
\end{align*}
Then $\Omega^*=\R_+ \setminus \Q$ and using the notation of \citep[Proof of Lemma 1]{bfhmo} the standard separators are given by $\xi_{0, A_0}=(0,0,1)$ and $\xi_{0, A_1}=(-1,0,0)$. Next we define
\begin{align*}
g(S^1, S^2, S^3)=
\left\{
			\begin{array}{ll}
					(\Delta S^2_1 + |\Delta S^1_1|+ |\Delta S^3_1|)^-&\text{ if } \Delta S^2_1 \le 0, \\
					(\Delta S^2_1 - |\Delta S^1_1|- |\Delta S^3_1|)^+&\text{ if } \Delta S^2_1 > 0, \\
					\end{array}
					\right.
\end{align*}
We note that for $\Delta S^1_1= \Delta S^3_1=0$ we have $g(S)=\Delta S_1^2$. So in particular to hedge $g$ on $\Omega^*$ we need initial capital $\pi_{\Omega^*}(g)=0$ and any hedging strategy satisfies $H_1=(H^1_1, 1, H^3_1)$ for $H^1_1, H^3_1\in \R$. For any such strategy to also superhedge on $\Q \cap [1/2, \infty)$ with initial capital $1/2$, $H_1^1$ has to satisfy in particular
\begin{align*}
1/2+H^1_1 \Delta S_1^1 -2H_1^2\ge 0 \hspace{0.5cm} \text{ for } \Delta S_1^1\le -2,
\end{align*}
so $H^1_1 \le -3/4$ as $H_1^2 \in [1,3/2]$. Lastly extending superhedging g on $\Q \cap [0,1/2)$ gives the constraint
\begin{align*}
1/2+2H^1_1+ H^3_1 \Delta S^3_1 \ge 0.
\end{align*}
Taking $\Delta S^3_1 = 0$ gives $H^1_1\ge -1/4$, a contradiction. Thus we will assume that all one-point arbitrages can be reduced to a single standard separator.
\end{Ex}

\begin{figure}[h]
\centering

\begin{tikzpicture}[scale=6,tdplot_main_coords]

\draw[thick,->] (1,0,0) -- (-1,0,0) node[anchor=north
west]{$\Delta S^1_{1}$}; \draw[thick,->] (0,0.8,0) -- (0,-0.8,0)
node[anchor=south]{$\Delta S^2_{1}$}; \draw[thick,->] (0,0,-0.5) --
(0,0,0.7) node[anchor=south]{$\Delta S^3_{1}$};

\fill [color=blue] (0,0.2,0) circle (0.5pt); \fill [color=blue]
(0,0.15,0) circle (0.5pt); \fill [color=blue] (0,0.1,0) circle
(0.5pt); \fill [color=blue] (0,0.05,0) circle (0.5pt); \fill
[color=blue] (0,0.0,0) circle (0.5pt); \fill [color=blue]
(0,-0.05,0) circle (0.5pt); \fill [color=blue] (0,-0.1,0) circle
(0.5pt); \fill [color=blue] (0,-0.15,0) circle (0.5pt); \fill
[color=blue] (0,-0.2,0) circle (0.5pt); \fill [color=blue]
(0,-0.25,0) circle (0.5pt); \fill [color=blue] (0,-0.3,0) circle
(0.5pt); \fill [color=blue] (0,-0.35,0) circle (0.5pt); \fill
[color=blue] (0,-0.4,0) circle (0.5pt); \fill [color=blue]
(0,-0.45,0) circle (0.5pt); \fill [color=blue] (0,-0.5,0) circle
(0.5pt); \fill [color=blue] (0,-0.55,0) circle (0.5pt); \fill
[color=blue] (0,-0.6,0) circle (0.5pt); \fill [color=blue]
(0,-0.65,0) circle (0.5pt); \fill [color=blue] (0,-0.7,0) circle
(0.5pt);
\fill [color=blue] (0,-0.75,0) circle
(0.5pt);

\fill [color=red] (-0.2,0,0) circle (0.5pt);
\fill [color=red] (-0.2,0,0.05) circle (0.5pt);
\fill [color=red] (-0.2,0,0.1) circle (0.5pt);
\fill [color=red] (-0.2,0,0.15) circle (0.5pt);
\fill [color=red] (-0.2,0,0.2) circle (0.5pt);
\fill [color=red] (-0.2,0,0.25) circle (0.5pt);
\fill [color=red] (-0.2,0,0.3) circle (0.5pt);

\fill [color=green] (0,0.2,0) circle (0.5pt);
\fill [color=green] (0.05,0.2,0) circle (0.5pt);
\fill [color=green] (0.1,0.2,0) circle (0.5pt);
\fill [color=green] (0.15,0.2,0) circle (0.5pt);
\fill [color=green] (0.2,0.2,0) circle (0.5pt);
\fill [color=green] (0.25,0.2,0) circle (0.5pt);
\fill [color=green] (0.3,0.2,0) circle (0.5pt);
\fill [color=green] (0.35,0.2,0) circle (0.5pt);
\fill [color=green] (0.4,0.2,0) circle (0.5pt);
\fill [color=green] (0.45,0.2,0) circle (0.5pt);
\fill [color=green] (0.5,0.2,0) circle (0.5pt);
\fill [color=green] (0.55,0.2,0) circle (0.5pt);
\fill [color=green] (0.6,0.2,0) circle (0.5pt);
\fill [color=green] (0.65,0.2,0) circle (0.5pt);
\fill [color=green] (0.7,0.2,0) circle (0.5pt);
\fill [color=green] (0.75,0.2,0) circle (0.5pt);
\fill [color=green] (0.8,0.2,0) circle (0.5pt);
\fill [color=green] (0.85,0.2,0) circle (0.5pt);
\fill [color=green] (0.9,0.2,0) circle (0.5pt);
\fill [color=green] (0.95,0.2,0) circle (0.5pt);
\fill [color=green] (1,0.2,0) circle (0.5pt);

\fill [color=blue] (4.58,3.2,3.48) circle (0.5pt);
\fill [color=green] (4.58,3.2,3.36) circle (0.5pt);
\fill [color=red] (4.58,3.2,3.24) circle (0.5pt);
\draw (4.5,3,3.55) node[anchor=north west] {Legend:};
\draw (4.5,3.2,3.5) node[anchor=north west] { $\mathbb{R}^+\setminus\mathbb{Q}$};
\draw (4.5,3.2,3.39) node[anchor=north west] { $\mathbb{Q} \cap [1/2,\infty)$};
\draw (4.5,3.2,3.26) node[anchor=north west] { $\mathbb{Q} \cap [0,1/2)$};

\end{tikzpicture}
\caption[Example \ref{Ex. joh1}]{$\Omega\subseteq \R^3$ in Example \ref{Ex. joh1}}
\end{figure}
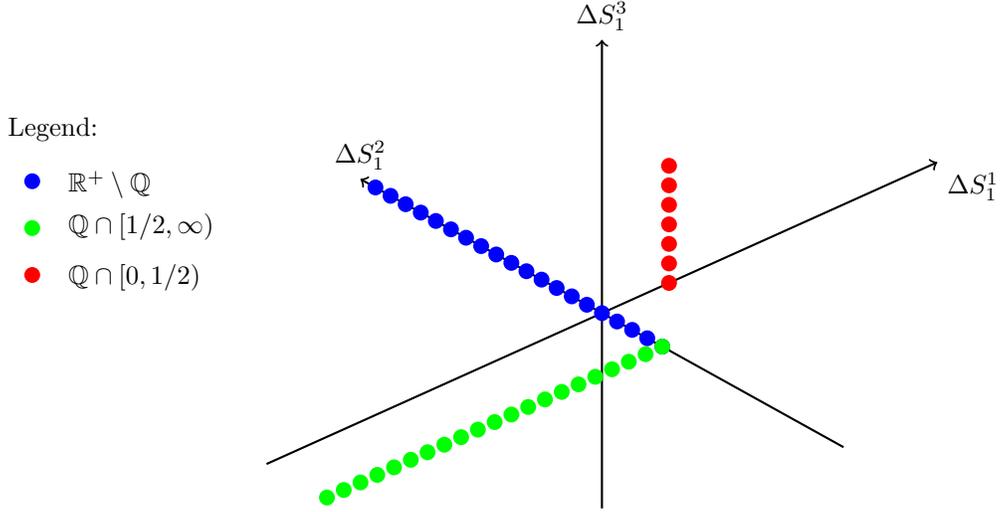

Note that Example \ref{Ex. joh1} can be easily altered to make $\Delta S_{t+1}(\Omega \cap \Sigma_t^\omega)$ compact by adding additional points. For clarity of exposition we have refrained from doing this but we conclude that Assumptions \ref{Ass. joh1} and \ref{Ass. joh2} are not sufficient for $d > 2$. We have to add a last assumption, which guarantees measurability of the corresponding Universal Arbitrage Aggregator. Intuitively it states, that locally, i.e., for every $0 \le t \le T-1$ and every $\omega \in \Omega$, there exists at most one arbitragable direction of the evolution of assets $S$, so that the first standard separator is already the Universal Arbitrage Aggregator:

\begin{Ass}\label{Ass. joh3}
For all $\omega \in \Omega$ and $0 \le t \le T-1$ we have $\xi_{t+1,\Omega \cap \Sigma_t^{\omega}}=H^*_t$, where $H^*$ is the Universal Arbitrage Aggregator of \citep{bfhmo} for the set $\Sigma_t^{\omega} \cap \Omega$.
\end{Ass}

\begin{Theorem} Suppose that $\Phi=0$ and $X\ni \omega \mapsto S_t(\omega)$ is continuous for all $1\leq t\leq T$. For an analytic $\Omega \subseteq X$ satisfying Assumptions \ref{Ass. joh1}, \ref{Ass. joh2} and \ref{Ass. joh3} the Superhedging Duality of \citep{bfhmo} extends from $\Omega^*$ to $\Omega$ for all continuous $g: X \to \R$, i.e., we have
\begin{align*}
\sup_{\Q \in \mathcal{M}_{\Omega}^f} \E_{\Q}(g)&=\inf\{x \in \R \ | \ \exists H \in \mathcal{H}(\mathbb{F}^{\mathcal{U}}) \text{ such that } x+H \circ S_T \ge g \text{ on } \Omega^* \}\\&= \inf\{x \in \R \ | \ \exists H \in \mathcal{H}(\mathbb{F}^{\mathcal{U}}) \text{ such that } x+H \circ S_T \ge g \text{ on } \Omega \}  
\end{align*}
\end{Theorem}

\begin{proof}
As before, we prove the claim by backward induction over $t=0, \dots, T-1$. Let us now fix $\omega \in \Omega$. We assume $\Sigma_t^{\omega}\cap \Omega^*\neq \emptyset$ and $(\Sigma_t^{\omega}\cap \Omega) \setminus (\Sigma_t^{\omega}\cap \Omega^*) \neq \emptyset$, otherwise the claim is trivial.  We first look at the case, where $\text{proj}_{\Delta S_{t+1}(\Sigma_t^{\omega} \cap \Omega^*)}(\Delta S_{t+1}(v))$ is an element of $\Delta S_{t+1}(\Sigma_{t+1}^{\omega}\cap \Omega^*)$, i.e., there exists $v' \in \Sigma_t^{\omega}\cap \Omega^*$ such that $\text{proj}_{\Delta S_{t+1}(\Sigma_t^{\omega} \cap \Omega^*)}(\Delta S_{t+1}(v))= \Delta S_{t+1}(v')$ . Note that by Assumption \ref{Ass. joh3} the standard separator $\xi_{t+1, \Omega\cap \Sigma_t^{\omega}}$ is orthogonal to $\text{span} (\Delta S_{t+1}(\Sigma_{t}^{\omega}\cap \Omega^*))$. By definition of the superhedging price on $\Omega$ there exists an $\mathcal{F}_{t}^{\mathcal{U}}$-measurable strategy $H_{t+1}$ such that
\begin{align*}
\pi_{t,\Omega^*}(g)(v')+H_{t+1}(\omega)\Delta S_{t+1}(v') \ge \hat{\pi}_{t+1}(g)(v') \hspace{0.5cm} \text{for all } v' \in \Sigma_t^{\omega} \cap \Omega^*,
\end{align*}
where we can assume without loss of generality that $H_{t+1}(\omega) \in \text{span}(\Delta S_{t+1}(\Sigma_{t}^{\omega}\cap \Omega^*))$. Now we fix $v \in \Sigma_t^{\omega}\cap \Omega$ and $v'$ the corresponding orthogonal projection. Let $\epsilon >0$. As $\pi_{t+1,\Omega^*}(g)$ is uniformly continuous on $\Omega\cap \Sigma_t^{\omega}$, we can use \cite[Theorem 1]{vanderbei1991uniform} (in connection with Tietze's extension theorem to extend the domain to a convex set) in order to find $\delta >0$ such that
\begin{align*}
\epsilon+\pi_{t+1,\Omega^*}(g)(v')+ \frac{S_{t+1}(v)-S_{t+1}(v')}{\delta/ \epsilon} \xi_{t+1,\Sigma_t^{\omega} \cap \Omega}(v) \ge   \pi_{t+1,\Omega^*}(g)(v),
\end{align*}
where $\delta$ is chosen such that for all $w, \tilde{w} \in \Sigma_{t}^{\omega}\cap \Omega$ we have $|\pi_{t+1,\Omega^*}(g)(w)-\pi_{t+1,\Omega^*}(g)(\tilde{w})| \le \epsilon$ whenever $|S_{t+1}(w)-S_{t+1}(\tilde{w}))|<\delta$. Note that $\Delta S_{t+1}(v)-\Delta S_{t+1}(v')$ is orthogonal to $H_{t+1}(\omega)$ and
\begin{align*}
&\epsilon+\pi_{t,\Omega^*}(g)(g)(v)+\left(H_{t+1}(\omega)+\frac{\xi_{t+1,\Sigma_t^{\omega}\cap \Omega}(v)}{\delta/\epsilon}\right)(\Delta S_{t+1}(v')+\Delta S_{t+1}(v)-\Delta S_{t+1}(v'))\\
\ge &\ \epsilon+\pi_{t+1,\Omega^*}(g)(v')+ \frac{S_{t+1}(v)-S_{t+1}(v')}{\delta/ \epsilon} \xi_{t+1,\Sigma_t^{\omega} \cap \Omega}(v)
\ge  \pi_{t+1,\Omega^*}(g)(v),
\end{align*}

Next we use the assumption that $A_t^{\omega}$ is bounded and has no points of convergence in $\text{span}(\Delta S_{t+1}(\Sigma_t^{\omega}\cap \Omega^*))$ outside the set $\Delta S_{t+1}(\Sigma_t^{\omega}\cap \Omega^*)$. In particular the continuous functions $\pi_{t+1,\Omega^*}(g)$ and $H_{t+1}\Delta S_{t+1}$ are bounded on $A_t^{\omega}$. There exists $\delta >0$ such that for all $v \in \Delta S_{t+1}(\Sigma_t^{\omega} \cap \Omega^*),$ $\tilde{v} \in A_t^{\omega}$ with $|v-\tilde{v}|< \delta$ we still have 
\begin{align*}
\epsilon+\pi_{t,\Omega^*}(g)(\omega)+H_{t+1}(\omega)\Delta S_{t+1}(\tilde{v}) \ge \pi_{t+1,\Omega^*}(g)(\tilde{v}).
\end{align*}
By assumption there exists $\tilde{\delta}>0$ such that $\text{dist}(\tilde{v}, \text{span}(\Delta S_{t+1}(\Sigma_t^{\omega}\cap \Omega^*)))>\tilde{\delta}$ for all $\tilde{v} \in A_t^{\omega}$ with $\text{dist}(\tilde{v}, \Delta S_{t+1}(\Sigma_t^{\omega}\cap \Omega^*))>\delta$. Define $\pi_{\max}=\sup_{v \in A_t^{\omega}} \pi_{t+1,\Omega^*}(g)(v)<\infty$ and $C=\inf_{v \in A_t^{\omega}} H_{t+1}(\omega) \Delta S_{t+1}(v)+\pi_{t,\Omega^*}(g)(\omega)>-\infty$. Now we note that
\begin{align*}
\epsilon+\pi_{t+1,\Omega^*}(g)(\tilde{v})+H_{t+1}(\tilde{v})\Delta S_{t+1}(\tilde{v})+ \frac{|\pi_{\max}|+|C|}{\tilde{\delta}}\xi_{t+1, \Sigma_{t}^{\omega}\cap \Omega}(v) \Delta S_{t+1}(\tilde{v})  \ge\pi_{t+1,\Omega^*}(g)(\tilde{v})
\end{align*}
for all $\tilde{v} \in A_t^{\omega}$. This concludes the proof.
\end{proof}

%

\subsection{Comparison of Strong and Uniformly Strong Arbitrage}\label{Sec. compar}
We take now a closer look at the notions $\textbf{SA}(\Omega)$ and $\textbf{USA}(\Omega)$ and establish their equivalence in specific market setups. Clearly every Uniformly Strong Arbitrage is a Strong Arbitrage. In general the opposite assertion is not true: take for example $d=1$, $S_0=1$, $S_1(\omega)=\omega$, $\Omega=(1,2]$, then every $H_1>0$ is a Strong Arbitrage, but there do not exist any Uniformly Strong Arbitrages. This simple example can be generalised: a one-period market in the canonical setting with $S_0=1$ and an open convex set $\Omega$ such that $\{1\} \cap \Omega =\emptyset$ and $1 \in \bar{\Omega}$ admits a Strong Arbitrage but exhibits no Uniformly Strong Arbitrages. On the level of superhedging prices a Uniformly Strong Arbitrage on $\Omega$ corresponds to $\pi_{\Omega}(0)=-\infty$. For a financial market which exhibits a Strong Arbitrage but no uniformly Strong Arbitrages, the Pricing-Hedging duality cannot hold (as there are no martingale measures supported on $\Omega$) but $\pi_{\Omega}(0)=0$. In conclusion, the difference between Uniformly Strong Arbitrage and Strong Arbitrage can be seen as a property describing the boundary of the prediction set $\Omega$ and thus manifests itself in the boundary behaviour of the superhedging functional $$S_0 \mapsto \inf\{ x \in \R \ | \ \exists H \in \R^d \text{ s.t. } x+H (S_1-S_0) \ge 0 \text{ on } \Omega\}.$$ As it is an upper semicontinuous function, it takes the value zero on the boundary of $\Omega$, while its lower semicontinuous version takes the value $-\infty$. Nevertheless the two notions agree in specific cases, which we now explore.\\
We assume the canonical setting $X_1=\R_+^d$, $S_0(\omega)=s_0$ and set $\mathcal{F}_t=\mathcal{F}^\nf_t$ for all $0\le t\le T$. In this section we allow for countably many statically traded options, $\Lambda=\mathbb{N}$, but only of European type, $\Phi=\{\phi_n=\phi_n(S_T)\ | \ n \in \N\}$, with real-valued continuous payoffs and a common maturity $T$. We write $c_{00}=c_{00}(\Nb)$ for simplicity. 
We fix a closed subset $\Omega \subseteq (\R_+^d)^T$ and recall that martingale measures on $\Omega$ calibrated to $\Phi$ are denoted by $\mathcal{M}_{\Omega,\Phi}(\mathbb{F})$. We define $|S(\omega)|_1\coloneqq \sum_{t=1}^T \sum_{k=1}^d |S_t^k(\omega)|$ and 
denote by $C^b_{|S|_1}(\Omega)$ the space of real-valued continuous functions $f:\Omega \mapsto \R$ such that $$\sup_{\omega\in \Omega} \frac{|f(\omega)|}{|S(\omega)|_1 \vee 1}<\infty.$$  Finally, we define the calibrated supermartingale measures as
\begin{align*}
\mathcal{SM}_{\Omega, \Phi}(\mathbb{F})\coloneqq \left\{\Q \in \mathcal{P}(\Omega) \ | \ \E_{\Q}[\phi_n]\le 0 \ \forall n \in \N, \ \E_{\Q}[S_t|\mathcal{F}_{t-1}]\le S_{t-1}\textrm{ a.s.} \ \forall t\leq T\right\}.
\end{align*}
The following theorem can be seen as a unification of \citep[Theorem 1.3]{acciaio2013model}, \citep[Prop.\ 2.2, p.6]{cox2011robusta} and \citep[Cor.\ 4.6]{bartl2017duality2}. We also refer to \citep[Thm.\ 3]{bfhmo}, who extend \citep[Thm.\ 1.3]{acciaio2013model} under the assumption $\Omega= \Omega^*$ and to \citep[Thm.\ C.5]{riedsonbuz17} for a general discussion in the case $\Phi=0$. In contrast to the work of \cite{{acciaio2013model}}, we do not need to assume the existence of a convex superlinear payoff $g$, which might be artificial in some settings, but explicitly enforce tightness of martingale measures through the $\textbf{WFLVR}(\Omega)$ condition.
\begin{Theorem}
The following hold:
\begin{enumerate} 
\item \textbf{SA}$(\Omega)$ $\Leftrightarrow$ $\textbf{USA}(\Omega)$.
\item Assume $\phi_n \in C^b_{|S|_1}(\Omega)$, no short-selling in any of the assets and $$\lim_{|S_T|_1 \to \infty}\frac{(\phi_{n}(S_T))^{-}}{|S_T|_1}=0$$ for all $n \in \N$. Then
\begin{align*}
&\textbf{SA}(\Omega)\Leftrightarrow \textbf{USA}(\Omega) \Leftrightarrow \textbf{WFLVR}(\Omega)\Leftrightarrow
\mathcal{SM}_{\Omega, \Phi}(\mathbb{F})= \emptyset.
\end{align*}
\item As in \textit{(2)} assume that $\phi_n \in C^b_{|S|_1}(\Omega)$. Furthermore assume that for every sequence $(\omega_n)_{n\in \N}$ with $\lim_{n \to \infty}|S(\omega_n)|_1=\infty$, there exists a sequence $(h^k,H^k)_{k \in \N}$ of trading strategies, a constant $C>0$ and a sequence $(p^k)_{k\in \N}$ such that
\begin{itemize}
\item $\lim_{k \to \infty} \lim_{n \to \infty} \frac{(h^k\cdot \Phi+H^k\circ S_T)(\omega_n)}{|S(\omega_n)|_1\vee 1}>0$ and
\item  $|h^k\cdot \Phi+H^k\cdot S_T| \le C(|S|_1\vee 1)$ on $\Omega$ for all $k\in \N$,
\item $\lim_{k\to \infty}(h^k\cdot \Phi+H^k\circ S_T)(\omega)=-\lim_{k\to \infty}p^k$ for all $\omega \in \Omega$.
\end{itemize}
Then $$\textbf{SA}(\Omega)\Leftrightarrow \textbf{USA}(\Omega) \Rightarrow \textbf{WFLVR}(\Omega) \Leftrightarrow \mathcal{M}_{\Omega, \Phi}(\mathbb{F})= \emptyset,$$
but in general $\textbf{WFLVR}(\Omega)$ does not imply $\textbf{SA}(\Omega)$.
\end{enumerate}
\end{Theorem}

\begin{Rem}
\begin{enumerate}
\item The case $|\Phi|<\infty$ is covered in \citep{bouchard2015arbitrage,bfhmo}), while the case $|\Phi|=\infty$ is not. The basic idea in both works is to inductively construct a martingale measure calibrated to a finite number of options.
\item Contrary to the case $|\Phi|<\infty$ (see \citep[proof of Theorem 1, p.1050]{bfhmo}), the set $\mathcal{M}_{\Omega,\Phi}(\mathbb{F})$ might not necessarily contain any finitely supported martingale measures.
\item An example showing that $\textbf{WFLVR}(\Omega)$ does not imply $\textbf{SA}(\Omega)$ is given in \citep[Prop. 2.2]{cox2011robusta}.
\item A special but important case of \textit{(3)} is $T=1,d=1$ and $\phi_n(S_1)=(S_1-n)^+-p^n$, where $p^n\ge 0$. In this case we can set $H^k=0$, $h^k=e_k$ for all $k\in \N$, where $e_k$ is the $k$th unit vector and note that $$\lim_{k \to \infty} \lim_{n \to \infty} \frac{(h^k\cdot \Phi+H^k(S_1-S_0))(\omega_n)}{|S(\omega_n)|_1\vee 1}=\lim_{k \to \infty} \lim_{n \to \infty} \frac{(S_1(\omega_n)-k)^+-p^k}{|S(\omega_n)|_1}=1>0$$ and
$$\lim_{k\to \infty} (S_1(\omega)-k)^+-p^k=-\lim_{k\to \infty} p^k,$$ in particular all three conditions in \textit{(3)} are satisfied.
\end{enumerate}
\end{Rem}

\begin{proof}\label{proof:1}
For simplicity of exposition we only give the proof for $T=1$. This conveys the important ideas, while the multiperiod case extends these via a dynamic programming approach and can found in \citep{WieselDPhil}.
\\
Regarding \textit{(1)}, clearly $\textbf{USA}(\Omega)\Rightarrow\textbf{SA}(\Omega)$, so we show $\textbf{SA}(\Omega)\Rightarrow\textbf{USA}(\Omega)$. Let $(h,H)\in \R^k \times \R^d$ be a Strong Arbitrage. We show that it is actually a Uniformly Strong Arbitrage. For $x\in \R^d_+$ we denote by $|x|_1\coloneqq \sum_{i=1}^d x^i$ the $\ell_1$-norm of $x$ and define the compact set $K=[0,s_0^1+2|s_0|_1] \times[0,s_0^2+2|s_0|_1]\times \dots \times [0,s_0^d+2|s_0|_1]$. Then, as $S_1 \mapsto h\cdot \Phi(S_1)+H(S_1-S_0)$ is continuous and positive on a compact set $\Omega \cap K$, there exists $\epsilon>0$ such that
\begin{align*}
h\cdot \Phi(S_1)+H\cdot(S_1-S_0) \ge \epsilon \quad \text{on }K \cap \Omega.
\end{align*}
Scaling $(h,H)$ suitably we can without loss of generality assume take $\epsilon=2|s_0|_1$. Let $\textbf{e}=(1, \dots,1)$ be the row unit vector in $\R^d$. Then
\begin{align}\label{eq. trivial}
h\cdot \Phi(S_1)+(H+\textbf{e})\cdot (S_1-S_0) \ge 2|s_0|_1-|s_0|_1=|s_0|_1 \quad \text{on }K \cap \Omega.
\end{align}
Furthermore on $\Omega \setminus K$ we have 
\begin{align*}
h\cdot \Phi(S_1)+(H+\textbf{e})\cdot (S_1-S_0) \ge \textbf{e}\cdot (S_1-S_0)\ge 2|s_0|_1-|s_0|_1=|s_0|_1.
\end{align*}
Now we show \textit{(2)}. Clearly the relation $\textbf{USA}(\Omega)\Rightarrow\textbf{WFLVR}(\Omega)$ holds and by  \textit{(1)} also $\textbf{SA}(\Omega)\Leftrightarrow\textbf{USA}(\Omega)$. Further, $\textbf{WFLVR}(\Omega)$ readily implies $\mathcal{SM}_{\Omega, \Phi}(\mathbb{F})= \emptyset$ since otherwise if $\Q\in\mathcal{SM}_{\Omega, \Phi}(\mathbb{F})$ then, by Fatou's lemma, 
\begin{align*}
0\ge\liminf_{n \to \infty} \E_{\Q}[h^n \cdot \Phi+H^n (S_1-S_0)] \ge \E_{\Q}[\liminf_{n \to \infty} h^n \cdot \Phi+H^n (S_1-S_0)]>0,
\end{align*}
a contradiction. Next we show $\textbf{NUSA}(\Omega) \Rightarrow (\mathcal{SM}_{\Omega, \Phi}(\mathbb{F})\neq \emptyset) $ following closely the argument in \citep[proof of Prop. 2.3 and Theorem 1.3, pp.\ 240-242]{acciaio2013model}. We denote by $c^+_{00}$ the subset of all non-negative sequences in $c_{00}$. We define the set 
\begin{align*}
K:= \left\{h \cdot \Phi(S_1)+H(S_1-S_0) \ \big| \ (h,H) \in c_{00}^+ \times \R^d_+\right\} \subseteq C^b_{|S|_1}(\Omega).
\end{align*}
Note that $K$ is convex and non-empty. Furthermore denote the positive cone of  $C^b_{|S|_1}(\Omega)$ by
\begin{align*}
C_{++}(\Omega)= \left\{f \in C^b_{|S|_1}(\Omega)  \ \bigg| \ \inf_{\omega \in\Omega} \frac{f(\omega)}{|S(\omega)|_1 \vee 1}>0\right\}.
\end{align*}
By $\textbf{NUSA}(\Omega)$ we have $K \cap C_{++}(\Omega)=\emptyset$. An application of Hahn-Banach theorem yields existence of a positive measure $\mu=\mu^r+\mu^s$ such that
\begin{align*}
\int_{\Omega} \frac{f}{|S|_1 \vee 1}\,d\mu &>0 \quad \text{for all }f \in C_{++}(\Omega),\\
\int_{\Omega} \frac{f}{|S|_1 \vee 1}\,d\mu&\le 0 \quad \text{for all }f \in K.
\end{align*}
We now aim to show that the normalised measure $\Q$ given by  $$d\Q:=\frac{1}{|S|_1 \vee 1} \left(\int \frac{1}{|S|_1 \vee 1}\,d\mu^r \right)^{-1}d\mu^r$$ is an element of $\mathcal{SM}_{\Omega, \Phi}$. For this let us first assume that $\mu^r=0$. Then $$\int_{\Omega} \frac{ \textbf{e} (S_1-S_0)}{|S|_1 \vee 1}\,d\mu=\int_{\Omega} \frac{|S|_1-|S_0|_1}{|S|_1 \vee 1}\, d\mu^s=\int_{\Omega} 1 \,d\mu^s>0$$
as $\mu$ is positive, which is a contradiction.
As $\int_{\Omega} \frac{(\phi_n(S_1))^{-}}{|S|_1 \vee 1}\,d\mu^s=0$, we conclude
\begin{align*}
\int_{\Omega} \frac{\phi_n(S_1)}{|S|_1 \vee 1} \, d\mu^r \le \int_{\Omega} \frac{\phi_n(S_1)}{|S|_1 \vee 1} \,d\mu \le 0 \quad \text{ for all }n \in \N.
\end{align*}
Furthermore
\begin{align*}
\int_{\Omega} \frac{S_1-S_{0}}{|S|_1 \vee 1} \, d\mu^r = 0 
\end{align*}
and thus $\textbf{NUSA}(\Omega)\Rightarrow \mathcal{SM}_{\Omega, \Phi}\neq \emptyset $ follows. \\
Lastly we show \textit{(3)}. For this we follow the same construction as in \textit{(2)}. In particular redefining 
\begin{align*}
K:= \bigg\{h \cdot \Phi(S_1)+H(S_1-S_0) \ \bigg| \ (h,H) \in c_{00} \times \R^d\bigg\} \subseteq C^b_{|S|_1}(\Omega).
\end{align*} we note that again by $\textbf{NUSA}(\Omega)$ we have $K \cap C_{++}(\Omega)=\emptyset$. Thus all that is left to show is $\mu^s=0$.
Let us assume towards a contradiction $\mu^s\ne 0$ and take $(h^k,H^k)_{k\in \N}$ such that 
\begin{align}\label{eq:limit}
\lim_{k \to \infty}\int_{\Omega} \frac{h^k\cdot \Phi(S_1)+H^k(S_1-S_0)}{|S|_1 \vee 1}d\mu^s>0.
\end{align}
Then by symmetry of $K$ and the same reasoning as in \textit{(2)} we have
\begin{align}\label{eq. calls}
\int_{\Omega} \frac{h^k\cdot \Phi(S_1)+H^k(S_1-S_0)}{|S|_1 \vee 1}\,d\mu=0 \quad \text{ for all }k\in \N.
\end{align}
Using \eqref{eq:limit} and \eqref{eq. calls}
\begin{align}\label{eq:new}
\lim_{k \to \infty}\int_{\Omega} \frac{h^k\cdot \Phi(S_1)+H^k( S_1-S_0)}{|S|_1 \vee 1}d\mu^r=-\lim_{k \to \infty}\int_{\Omega} \frac{h^k\cdot \Phi(S_1)+H^k(S_1-S_0)}{|S|_1 \vee 1}d\mu^s<0.
\end{align}
Note that for a sequence 
$(p^k)_{n\in \N}$ with
$$\lim_{k\to \infty}h^k\cdot \Phi(S_1)+H^k(S_1-S_0)=-\lim_{k\to \infty}p^k$$ for all $\omega \in \Omega$ we need to have by no \textbf{WFLVR}($\Omega$) that $\lim_{k\to \infty}p^k=0$, so the LHS of \eqref{eq:new} is equal to zero, a contradiction. 
\end{proof}

\section{Technical results and proofs} \label{sec:proofs}

\subsection{Proof of Theorems \ref{Thm. bn vs bfhmo} and \ref{Thm. S}}\label{subseq. ftap}

We start with the following technical observation:
\begin{Prop}\label{prop 2}
Let $\Omega$ be analytic. Then the FTAP of \citep{bouchard2015arbitrage} implies:
\begin{align*}
\textrm{\textbf{N1pA}}(\Omega,\mathbb{F}^{\mathcal{U}}) \Leftrightarrow \Omega= \Omega^*_{\Phi}
\end{align*}
\end{Prop}
\begin{proof}
Set $\hat{\mathfrak{P}}\coloneqq \mathcal{P}^f(\Omega)$. To apply the FTAP of \citep{bouchard2015arbitrage} we only need to show that $\hat{\mathfrak{P}}_t(\omega)\coloneqq \mathcal{P}^f(\text{proj}_{t+1}(\Omega \cap \Sigma_t^{\omega}))$ has analytic graph:
we therefore fix $n \in \N$ and consider the Borel measurable function
\begin{align*}
\Sigma\colon X^n \to X^n \times \left(\R_+^d\right)^{(t+1)n} \hspace{0.5cm} (\omega^1, \dots \omega^n) \mapsto (\omega^1, \dots, \omega^n, S_{0:t}(\omega^1), \dots, S_{0:t}(\omega^n))
\end{align*} 
and note that the image $\Sigma(\Omega^n)$
is analytic, since $\Omega$ is analytic and the image of an analytic set under a Borel measurable map as well as the Cartesian product of analytic sets is analytic (see \citep[Prop. 7.38 \& 7.40, p. 165]{bertsekas1978stochastic}). Next we consider the continuous function 
\begin{align*}
F\colon &X^n \times \left(\R_+^d\right)^{(t+1)} \to X^n \times \left(\R_+^d\right)^{(t+1)n} \\ 
&(\omega^1, \dots \omega^n, x) \mapsto (\omega^1, \dots, \omega^n, x, \dots ,x).
\end{align*}
Note that 
\begin{align*}
F\left(X^n \times \left(\R_+^d\right)^{(t+1)} \right) \cap \Sigma(\Omega^n)
\end{align*}
is analytic and as projections of analytic sets are analytic
\begin{align*}
A_n \coloneqq \{ (\omega, \tilde{\omega}_1, \dots ,\tilde{\omega}_n) \ | \ &\omega \in \Omega,\ \tilde{\omega}_i \in \text{proj}_{t+1}(\Omega \cap \Sigma_t^{\omega}), \ i=1,\dots, n\}
\end{align*}
is analytic as well. Let $\Delta_n \subseteq \R^n$ denote the simplex. Since the functions
\begin{align*}
G: \ &A_n \times \Delta_n \to \Omega \times \mathcal{P}(X_n)\times \Delta_n\\
&(\omega, \tilde{\omega}_1, \dots \tilde{\omega}_n,\lambda_1, \dots, \lambda_n) \mapsto (\omega, \delta_{\tilde{\omega}_1}, \dots, \delta_{\tilde{\omega}_n}, \lambda_1, \dots, \lambda_n) 
\end{align*}
and
\begin{align*}
H:\ &\Omega \times \mathcal{P}(X_n) \times \Delta_n \to \Omega \times \mathcal{P}(X_1)\\
&(\omega, \delta_{\tilde{\omega}_1}, \dots, \delta_{\tilde{\omega}_n}, \lambda_1, \dots, \lambda_n) \mapsto \left(\omega, \sum_{i=1}^n \delta_{\tilde{\omega}_i}\lambda_i\right)
\end{align*}
are continuous, it follows that graph$\left(\hat{\mathfrak{P}}_t\right)=\bigcup_{n \in \N}H(G(A_n \times \Delta_n))$ is analytic.\\
Take now $\omega \in \Omega$ and $\P \in \mathcal{P}^f(\Omega)$ such that $\P(\{\omega\})>0$. By the FTAP of \citep{bouchard2015arbitrage} there exists $\Q \in \mathcal{M}$ such that $\Q \ll \tilde{\P}$ for some $\tilde{\P}\in \mathcal{P}^f(\Omega)$, $\E_{\Q}[\phi_j]=0$ for all $j=1, \dots,k$ and $\P \ll \Q$. In particular $\Q \in \mathcal{M}^f_{\Omega}$ and $\Q(\{\omega\})>0$. \\
Lastly assume that $\Omega= \Omega_{\Phi}^*$ and fix $\P \in \hat{\mathfrak{P}}$ such that $\text{supp}(\P)=\{\omega_1, \dots, \omega_n\}$ for some $n \in \N$. We can find $\Q_1, \dots, \Q_n \in \mathcal{M}^f_{\Omega}$ such that $\Q_i(\{\omega_i\})>0$ for $i=1, \dots, n$. Then $\Q \coloneqq 1/n \sum_{i=1}^n \Q_i \in \mathcal{M}^f_{\Omega,\Phi}$ and $\Q(\{\omega_i\})>0$ for $i=1, \dots, n$, i.e., $\P \ll \Q$.
\end{proof}

We now give a complete proof of the quasi-sure FTAP in \citep{bouchard2015arbitrage} using results from \citep{bfhmo}. We first look at the case $\Phi=0$ and start with an auxiliary lemma:

\begin{Lem}\label{Lem seperator}
Let $t \in \{1, \dots, T\}$ and $\Omega \subseteq X_t$ be analytic. Then the conditional standard separator of \citep{bfhmo} denoted by $\xi_{t, \Omega}$ is $\mathcal{F}^{\mathcal{U}}_{t-1}$-measurable.
\end{Lem}
\begin{proof}
We shortly recall arguments from \citep{bfhmo}[proof of Lemma 1]: let us define the multifunction 
\begin{align*}
\psi_{t, \Omega}: \omega \in X \twoheadrightarrow \{\Delta S_t(\tilde{\omega}) \ | \ \tilde{\omega}\in \Sigma_{t-1}^{\omega} \cap \Omega\}\subseteq \R^d.
\end{align*}
Then $\psi_{t, \Omega}$ is an $\mathcal{F}_{t-1}^{\mathcal{U}}$-measurable multifunction. Indeed, for $O \subseteq \R^d$ open we have
\begin{align*}
\{\omega \in X \ | \ \psi_{t,\Omega}(\omega) \cap O \neq \emptyset \}= S_{0:t-1}^{-1}(S_{0:t-1}((\Delta S_t)^{-1}(O)\cap \Omega)).
\end{align*}
As $\Delta S_t$ is Borel measurable $(\Delta S_t)^{-1}(O) \in \mathcal{F}_t^\nf$. Also as intersections, projections and preimages of analytic sets are analytic (see \citep[Prop. 7.35 \& Prop. 7.40]{bertsekas1978stochastic}), we find that $\{\omega \in X \ | \ \psi_{t,\Omega}(\omega) \cap O \neq \emptyset \}$ is analytic and in particular $\mathcal{F}_{t-1}^{\mathcal{U}}$-measurable. Let $\mathbb{S}^d$ be the unit sphere in $\R^d$, then by preservation of measurability also the multifunction
\begin{align*}
\psi_{t,\Omega}^*(\omega) \coloneqq \{ H \in \mathbb{S}^d \ | \ H \cdot y \ge 0\ \text{for all } y \in \psi_{t,\Omega}(\omega) \}
\end{align*}
is $\mathcal{F}_{t-1}^{\mathcal{U}}$-measurable and closed-valued. Let $\{\xi_{t,\Omega}^n\}_{n \in \N}$ be its $\mathcal{F}_{t-1}^{\mathcal{U}}$-measurable Castaing representation. The conditional standard separator is then defined as
\begin{align*}
\xi_{t,\Omega}= \sum_{n=1}^{\infty} \frac{1}{2^n}\xi_{t, \Omega}^n.
\end{align*}
\end{proof}
\begin{Rem}
We recall that this separator has the property that it aggregates all one-dimensional One-point Arbitrages on $\Sigma_{t-1}^{\omega}\cap \Omega$ in the sense that
\begin{align*}
\{\omega \in X \ | \ \xi(\omega) \cdot \Delta S_t(\omega)>0 \} \subseteq \{\omega \in X \ | \ \xi_{t,\Omega}(\omega) \cdot \Delta S_t(\omega)>0 \}
\end{align*}
for every measurable selector $\xi$ of $\psi_{t,\Omega}^*$.
\end{Rem}

\begin{proof}[Proof of Theorem \ref{Thm. bn vs bfhmo} for $\Phi=0$] \label{proof:case phi=0}
We start by proving the first part of Theorem \ref{Thm. bn vs bfhmo}, i.e., we are given a set of measures $\mathfrak{P}$ satisfying $(APS)$ and we need to construct $\Omega=\Omega^{\mathfrak{P}}$ such that \textit{(1)-(3)} are equivalent. We define for $\omega \in X_{t-1}$
\begin{align*}
\tilde{\chi}_{\mathcal{F}^\nf_{t-1}}(\omega)= \bigcap \{ A \subseteq \R^d \text{ closed } \ | \ \P (\Delta S_{t}(\omega, \cdot) \in A)=1 \ \forall \ \P \in \mathfrak{P}_{t-1}(\omega) \}.
\end{align*}
Then $\tilde{\chi}_{\mathcal{F}^\nf_{t-1}}$ is closed valued and $\P(\Delta S_{t}(\omega, \cdot) \in \tilde{\chi}_{\mathcal{F}^\nf_{t-1}}(\omega))=1$ for all $\P \in \mathfrak{P}_{t-1}(\omega)$ and all $\omega \in X_{t-1}$. Evidently
\begin{align*}
\tilde{\chi}_{\F_{t-1}^\nf}(\omega)&= \{x \in \R^d \ | \ \forall \epsilon>0 \ \P(\Delta S_t(\omega,\cdot) \in B(x,\epsilon))>0 \text{ for some }\P \in \mathfrak{P}_{t-1}(\omega) \} \\
&= \overline{\bigcup_{\P \in \mathfrak{P}_{t-1}(\omega)} \text{supp}( \P\circ \Delta S_t(\omega,\cdot)^{-1})}.
\end{align*}
Also it follows from \citep[Lemma 4.3, page 840]{bouchard2015arbitrage}, that $\tilde{\chi}_{\mathcal{F}^\nf_{t-1}}$ is analytically measurable. We quickly repeat their argument: let us define
\begin{align*}
l:X_{t-1} \times \mathcal{P}(X_1) \to \mathcal{P}(\R^d) \hspace{0.5cm} l(\omega, \P)= \P \circ \Delta S_{t}(\omega, \cdot)^{-1}.
\end{align*}
Then $l$ is Borel measurable. Next we consider
\begin{align*}
\mathcal{R}: X_{t-1} \twoheadrightarrow \mathcal{P}(\R^d) \hspace{0.5cm} \mathcal{R}(\omega) \coloneqq l(\omega, \mathfrak{P}_{t-1}(\omega))= \{\P \circ \Delta S_{t}(\omega,\cdot)^{-1} \ | \ \P \in \mathfrak{P}_{t-1}(\omega) \}.
\end{align*}
Since its graph is analytic, it follows that for $O \subseteq \R^d$ open
\begin{align*}
\{ \omega \in X_{t-1} \ | \ \tilde{\chi}_{\mathcal{F}^\nf_{t-1}}(\omega) \cap O \neq \emptyset \}&= \{ \omega \in X_{t-1} \ | \ R(O)>0 \text{ for some } R \in \mathcal{R}(\omega) \}\\ &= \text{proj}_{X_{t-1}} \{(\omega, R) \in \text{graph}(\mathcal{R}) \ | \ R(O)>0 \}
\end{align*}
is analytic as $R \mapsto R(O)$ is Borel.\\
We also note that for $\epsilon >0$ the function $x \mapsto R(B_{\epsilon}(x))$ is continuous, so $(x, R) \mapsto R(B_{\epsilon}(x))$ is Borel and
\begin{align*}
\text{graph}(\tilde{\chi}_{\mathcal{F}^\nf_{t-1}})&= \{(\omega, x) \in (X_{t-1} \times \R^d) \ | \ x \in \tilde{\chi}_{\mathcal{F}^\nf_{t-1}}(\omega) \}\\&= \bigcap_{\epsilon \in \Q_+}\text{proj}_{X_{t-1} \times \R^d} \left(  \{(\omega,R,x) \in (\text{graph}(\mathcal{R}) \times \R^d) \ | \ R(B_{\epsilon}(x)) >0 \} \right)
\end{align*}
is analytic. Now we define
\begin{align*}
U=\{ \omega \in X_t \ | \ \Delta S_{t}(\omega) \in \tilde{\chi}_{\mathcal{F}^\nf_{t-1}}(\omega) \}.
\end{align*}
Then 
\begin{align*}
U= \text{proj}_{X_t} (\text{graph}(\Delta S_{t}) \cap \text{graph}(\tilde{\chi}_{\mathcal{F}^\nf_{t-1}}))
\end{align*}
is analytic and by Fubini's theorem $\P(U)=1$ holds for all $\P \in \mathfrak{P}$. We now set 
\begin{align*}
\Omega^{\mathfrak{P}}= \bigcap_{t=1}^T \left\{ \omega \in X_t \ | \ \Delta S_{t}(\omega) \in \tilde{\chi}_{\mathcal{F}^\nf_{t-1}}(\omega) \right\},
\end{align*}
which is again analytic and $\P(\Omega^{\mathfrak{P}})=1$ for all $\P \in \mathfrak{P}$.\\
Having defined $\Omega^{\mathfrak{P}}$ we can now begin to prove equivalence of \textit{(1)-(3)}. If \textit{(2)} holds then \textit{(3)} follows immediately by a contradiction argument,
so we now show the more involved implications $\textit{(3)}\Rightarrow \textit{(1)}$ and $\textit{(1)}\Rightarrow \textit{(2)}$. Let us start with the proof of $\textit{(3)}\Rightarrow \textit{(1)}$: we assume that there exists $\hat{\P} \in \mathfrak{P}$ such that $\hat{\P}\left(\Omega^{\mathfrak{P}} \setminus (\Omega^{\mathfrak{P}})^*\right)>0$. We want to find $H \in \mathcal{H}(\mathbb{F}^{\mathcal{U}})$ and $\tilde{\P}\in \mathfrak{P}$ such that $H \circ S_T \ge 0$ $\mathfrak{P}$-q.s and $\tilde{\P}(H \circ S_T>0)>0$. For this we take $t=T-1$ and assume that 
\begin{align*}
\hat{\P}\left(\{\omega \in \text{proj}_{0:T-1}(\Omega^{\mathfrak{P}}) \ | \ \text{there is a One-point Arbitrage on } \Sigma_{T-1}^{\omega}\cap\Omega^{\mathfrak{P}} \}\right)>0.
\end{align*}
Let us now fix $\omega \in \{\text{proj}_{0:T-1}(\Omega^{\mathfrak{P}}) \ | \ \text{there is a One-point Arbitrage on } \Sigma_{T-1}^{\omega}\cap\Omega^{\mathfrak{P}} \}$. 
Denote by $\xi_{T,\Omega^{\mathfrak{P}}}$ the $\mathcal{F}_{T-1}^{\mathcal{U}}$-measurable standard separator of Lemma \ref{Lem seperator}. Now we define for each $\P \in \mathfrak{P}_{T-1}(\omega)$ the push-forward of $\P$ as
\begin{align*}
\P_{\Delta S_{T}(\omega, \cdot)}(A) = \P(\Delta S_{T}(\omega, \cdot) \in A),
\end{align*}
where $A \in \mathcal{B}(\R^d)$. We note that by definition
\begin{align*}
\P_{\Delta S_{T}(\omega, \cdot)}\left(\tilde{\chi}_{\mathcal{F}^\nf_{T-1}}(\omega)\right)=1
\end{align*} 
holds for all $\P \in \mathfrak{P}_{T-1}(\omega)$. With a slight abuse of notation we recall the set
\begin{align*}
B^1(\omega) \coloneqq \{\omega' \in \text{proj}_T(\Sigma_{T-1}^{\omega}\cap\Omega^{\mathfrak{P}}) \ | \ \xi_{T,\Omega^{\mathfrak{P}}}(\omega) \cdot \Delta S_T(\omega,\omega') >0 \}
\end{align*}
from \citep[proof of Lemma 1, Step 1]{bfhmo} and note that for all $\P \in \mathfrak{P}_{T-1}(\omega)$
\begin{align*}
&\P\left(\{ \omega'\in \text{proj}_T(\Sigma_{T-1}^{\omega}\cap \Omega^{\mathfrak{P}}) \ | \ \xi_{T,\Omega^{\mathfrak{P}}}(\omega) \cdot \Delta S_T(\omega, \omega')>0 \} \right)\\
&\hspace{5cm}= \P_{\Delta S_T(\omega, \cdot)} (\{x \in \R^d \ | \ \xi_{T,\Omega^{\mathfrak{P}}}(\omega) \cdot x >0\})
\end{align*}
follows. Clearly the set $\{ x \in \R^d \ | \ \xi_{T,\Omega^{\mathfrak{P}}}(\omega) \cdot x > 0 \}$ is open in $\R^d$, thus by definition of $\tilde{\chi}_{\mathcal{F}^\nf_{T-1}}(\omega)$ there is a $\tilde{\P} \in \mathfrak{P}_{T-1}(\omega)$ such that 
\begin{align*}
\tilde{\P}_{\Delta S_T(\omega, \cdot)} (\{x \in \R^d \ | \ \xi_{T,\Omega^{\mathfrak{P}}}(\omega) \cdot x >0\})>0
\end{align*}
or there are no One-point Arbitrages on $\Sigma_{T-1}^{\omega} \cap \Omega^{\mathfrak{P}}$. To finish the proof of $\textit{(3)} \Rightarrow\textit{(1)}$ we need to select $\tilde{\P}$ in a measurable way and this follows by standard arguments: 
Define the correspondence $\Psi: \R^d \times X_{T-1} \twoheadrightarrow \mathcal{P}(X_1)$ by
\begin{align*}
\Psi(H,\omega)= \{ \P \in \mathfrak{P}_{T-1}(\omega) \ | \ \E_{\P}[H \cdot \Delta S_T(\omega, \cdot)]^+>0 \}.
\end{align*}
This function has analytic graph by arguments in \citep[proof of Lemma 3.4, p.11]{nutz2014utility}, so we can employ the Jankov-von-Neumann theorem (cf. \citep[Proposition 7.49, page 182]{bertsekas1978stochastic})  to find a universally measurable kernel 
\begin{align*}
\P'_{T-1}: \R^d \times X_{T-1} \to \mathcal{P}(X_1)
\end{align*}
such that $\P'_{T-1}(H,\omega) \in \mathfrak{P}_{T-1}(\omega)$ for all $(H,\omega) \in \R^d \times X_{T-1}$ and $\P_{T-1}(H,\omega) \in \Psi(H,\omega)$ on $\{\Psi(H,\omega)\neq \emptyset\}$.
Then also the kernel
\begin{align*}
\omega \mapsto \tilde{\P}_{T-1}(\omega) \coloneqq \P_{T-1}'(\xi_{T,\Omega^{\mathfrak{P}}}(\omega), \omega)
\end{align*}
is universally measurable. Defining $\tilde{\P}\coloneqq \hat{\P}|_{X_{T-1}}\otimes \tilde{\P}_{T-1}$, which is the product measure formed from the restriction of $\hat{\P}$ to $X_{T-1}$ and $\tilde{P}_{T-1}$ gives $\tilde{\P}(\xi_{T, \Omega^{\mathfrak{P}}} \cdot \Delta S_T>0)>0$. This proves $\textit{(3)}\Rightarrow\textit{(1)}$ by backward induction.\\
Lastly we show $\textit{(1)} \Rightarrow \textit{(2)}$: let us assume $\P((\Omega^{\mathfrak{P}})^*)=1$ for all $\P \in \mathfrak{P}$. Note that by the arguments given in the proof of $\textit{(3)}\Rightarrow\textit{(1)}$ this means that 
\begin{align*}
\bigcup_{t=1}^T\{\text{proj}_{0:t-1}(\Omega^{\mathfrak{P}}) \ | \ \text{there is a One-point Arbitrage on } \Sigma_{t-1}^{\omega}\cap\Omega^{\mathfrak{P}} \}
\end{align*}
is a $\mathfrak{P}$-polar set, so in particular $0 \in \text{ri}(\tilde{\chi}_{\mathcal{F}^\nf_{t-1}}(\omega))$ for all $t=1, \dots, T$ and $\mathfrak{P}$-q.e. $\omega \in X$. Here $\text{ri}(\tilde{\chi}_{\mathcal{F}^\nf_{t-1}}(\omega))$ denotes the relative interior of the convex hull of $\tilde{\chi}_{\mathcal{F}^\nf_{t-1}}(\omega)$. Let $\hat{\P} \in \mathfrak{P}$ be fixed. 
We define for an arbitrary $\P \in \mathfrak{P}$ and $\omega \in X_{t-1}$ the support of $\P_{t-1}(\omega) \circ \Delta S_{t}^{-1}(\omega, \cdot)$ conditioned on $\mathcal{F}^\nf_{t-1}$ as
\begin{align*}
\chi_{\mathcal{F}^\nf_{t-1}}^{\P}(\omega)=\{ x \in \R^d \ | \ \P_{t-1}(\omega)(\Delta S_{t}(\omega, \cdot) \in B_{\epsilon}(x))>0 \text{ for all } \epsilon > 0 \}.
\end{align*}
Using selection arguments which are explained below, we can now find measurable selectors $\P_{(0,1)}, \dots , \P_{(0,d)},\P_{(1,1)}, \dots, \P_{(T-1,d)}$ such that
\begin{align*}
\P_{(t,1)}(\omega), \dots , \P_{(t,d)}(\omega) \in \mathfrak{P}_t(\omega)
\end{align*}
and $\P_{(0,1)}, \dots, \P_{(T-1,d)}$ fulfil the following property: define
\begin{align*}
\tilde{\P}_t(\omega)= \frac{1}{d+1} \left(\hat{\P}_t(\omega) +\sum_{i=1}^d \P_{(t,i)}(\omega) \right) 
\end{align*}
 for $t=0, \dots, T-1$ and every $\omega \in X_{t}$. Then for $\tilde{\P}= \tilde{\P}_0 \otimes \dots \otimes \tilde{\P}_{T-1}$ we have 
\begin{align*}
0 \in \text{ri}\left(\chi^{\tilde{\P}}_{{\mathcal{F}}^\nf_{t-1}}\right) \ \ \tilde{\P}\text{-a.s. for all }1 \le t \le T,
\end{align*} 
where $\text{ri}\left(\chi^{\tilde{\P}}_{{\mathcal{F}}^\nf_{t-1}}\right)$ denotes the relative interior of the convex hull of $\chi^{\tilde{\P}}_{{\mathcal{F}}^\nf_{t-1}}$.\\
We note that since $\mathfrak{P}_t(\omega)$ is convex, we have $\tilde{\P}_t(\omega) \in \mathfrak{P}_t(\omega)$ for $\omega \in X_{t}$ and by definition $\hat{\P} \ll \tilde{\P}$ holds. Now it follows from \citep[Theorem 1, page 1]{rokhlin2008proof}, that there exists a martingale measure $\Q$ equivalent to $\tilde{\P}$. The fact that  $\tilde{\P }\in \mathfrak{P}$ implies $\Q \in \mathcal{Q}_{\Pfr}$, which shows the claim.\\
We now present the measurable selection argument: we fix $t \in \{1, \dots, T\}$. Note that for all $\omega \in \text{proj}_{0:t-1}((\Omega^{\mathfrak{P}})^*)$ we conclude $0 \in \text{ri}(\Delta S_{t}(\omega,\Sigma_{t-1}^{\omega} \cap (\Omega^{\mathfrak{P}})^*))$ by definition of $(\Omega^{\mathfrak{P}})^*$, which implies by \citep[Theorem D, p.1]{bonnice1969relative} that there exist $\P^1, \dots, \P^d \in \mathfrak{P}_{t-1}(\omega)$, which might not be pairwise distinct, s.t.
\begin{align*}
0 \in \text{ri}\left(\text{supp}\left(\frac{\hat{\P}_{t-1}(\omega)+\P^1+ \dots+ \P^d}{d+1} \circ \Delta S_{t}(\omega,\cdot)^{-1}\right) \right).
\end{align*}
Note that $\omega \mapsto \hat{\P}_{t-1}(\omega)$ is universally measurable. 
We define the correspondence $\rho: \mathcal{P}(X_1)^{d+1} \twoheadrightarrow \R^d$ by
\begin{align*}
\rho: (\P^0, \P^1, \dots, \P^d) = \text{supp}\left(\frac{\P^0+ \P^1 + \dots + \P^d}{d+1} \circ \Delta S_{t}(\omega, \cdot)^{-1} \right).
\end{align*}
Note that for $O \subseteq \R^d$ open we have
\begin{align*}
\{(\P^0, \P^1, \dots, \P^d) \ | \ \rho(\P^0, \P^1, \dots, \P^d) \cap O \neq \emptyset \}&=\bigcup_{i=0}^d \{(\P^0, \P^1, \dots, \P^d )\ | \ \P^i \circ \Delta S_t(\omega,\cdot)^{-1}(O)>0 \}.
\end{align*}
Since $\P \mapsto \P(O)$ is Borel measurable, we conclude that $\rho$ is weakly measurable. Let us denote by $\mathbb{S}^d$ the unit sphere in $\R^d$. By preservation of measurability (cf. \citep[Exercise 14.12, page 653]{rockafellar2009variational}) it follows that the correspondence $\Psi: \mathcal{P}(X_1)^{d+1} \twoheadrightarrow \R^d$
\begin{align*}
\Psi(\P^0, \P^1, \dots, \P^d) = \{H \in \mathbb{S}^d \ | \ H \cdot y \ge 0 \text{ for all }y \in \rho(\P^0, \P^1, \dots, \P^d) \} 
\end{align*}
is weakly measurable. Then also the correspondence $\tilde{\Psi}:  \mathcal{P}(X_1)^{d+1} \twoheadrightarrow \R^d$
\begin{align*}
\tilde{\Psi}(\P^0, \P^1, \dots, \P^d) = &\{H \in \mathbb{S}^d \ | \ H \cdot y \le 0 \text{ for all }y \in \rho(\P^0,\P^1,\dots,\P^d) \} \\
&\cap \Psi(\P^0, \P^1, \dots, \P^d)
\end{align*}
is weakly measurable and closed-valued. Let $V$ be a countable base of $\R^d$. The set
\begin{align*}
&\{(\P^0, \P^1, \dots, \P^d) \ | \ \tilde{\Psi}(\omega, \P^1, \dots, \P^d)= \Psi(\P^0, \P^1, \dots, \P^d) \}\\
=&\bigcap_{O: O\in V} ( \{(\P^0, \P^1, \dots, \P^d) \ | \ \Psi \cap O \neq \emptyset \} \cap \{(\P^0, \P^1, \dots, \P^d) \ | \ \tilde{\Psi}\cap O \neq \emptyset \}\\  &\cup
\{(\P^0, \P^1, \dots, \P^d) \ | \ \Psi \cap O = \emptyset \} \cap \{(\P^0, \P^1, \dots, \P^d) \ | \ \tilde{\Psi}\cap O = \emptyset \} )
\end{align*}
is Borel measurable.
Note that for an arbitrary convex set $C \subseteq \R^d$ the relationship
\begin{align*}
0 \in \text{ri}(C) \Leftrightarrow (\forall\ H \in \mathbb{S}^d \text{ s.t. } H \cdot x \ge 0 \ \forall x \in C\ \Rightarrow\ H\cdot x =0 \ \forall x \in C)
\end{align*}
holds.
Let
\begin{align*}
A \coloneqq \{ (\P^0, \P^1, \dots, \P^d) \ | \ 0 \in \text{ri}\left(\rho(\P^0, \P^1, \dots \P^d)\right) \text{ for }i=1, \dots, d\}.
\end{align*}
Then from the above arguments it follows that $A$ is Borel and in particular the set-valued mapping
\begin{align*}
A(\omega, \P^0):=\{(\P^1,\dots,\P^d)\ | \ 0\in  \text{ri}\left(\rho(\P^0, \P^1, \dots \P^d)\right), \ \P^i\in \mathfrak{P}_{t-1}(\omega) \text{ for } i=1,\dots,d \}
\end{align*}
has analytic graph. We can now employ the Jankov-von-Neumann theorem (cf. \citep{bertsekas1978stochastic}, Proposition 7.49, page 182)  to find universally measurable kernels $\P_{t-1}^{i}: X_{t-1} \to \mathcal{P}(X_1)$ such that for every $\omega\in X_{t-1}$ we have $\P_{t-1}^{i}(\omega) \in \mathfrak{P}_{t-1}(\omega)$ and $$0\in \text{ri}\left(\rho(\hat{\P}_{t-1}(\omega), \P_{t-1}^1, \dots \P_{t-1}^d)\right).$$
This concludes the proof of $\textit{(1)}\Rightarrow\textit{(2)}$.\\
The second part of Theorem \ref{Thm. bn vs bfhmo} follows immediately from Proposition \ref{prop 2}.
\end{proof}

Before continuing the proof of Theorem \ref{Thm. bn vs bfhmo} let us first give a short remark on the measurability of the arbitrage strategies involved in the proof of above:

\begin{Rem}
By the FTAP of \citep{bfhmo} there exists a filtration $\tilde{\mathbb{F}}$ with $\mathbb{F}^\nf \subseteq \tilde{\mathbb{F}} \subseteq \mathbb{F}^M$ such that there is no Strong Arbitrage in $\mathcal{H}(\tilde{\mathbb{F}})$ on $\Omega^{\mathfrak{P}}$. More concretely there exists an $\mathcal{H}(\tilde{\mathbb{F}})$- and thus $\mathcal{H}(\mathbb{F}^M)$-measurable arbitrage aggregator $H^*$. So in particular if $\P(\Omega^{\mathfrak{P}} \setminus (\Omega^{\mathfrak{P}})^*)>0$ for some $\P \in \mathfrak{P}$, then $H^*$ is an $\mathcal{H}(\tilde{\mathbb{F}})$-measurable $\mathfrak{P}$-q.s arbitrage. In general the inclusion $\mathbb{\tilde{\mathbb{F}}} \subseteq \mathbb{F}^{\mathcal{U}}$ does not hold. This is why we need to construct a new $\mathbb{F}^{\mathcal{U}}$-measurable arbitrage strategy, which captures the arbitrages essential for $\mathfrak{P}$. More generally, in this paper we manage to avoid using projectively measurable sets, which were essential for the arguments in \citep{bfhmo}. In fact, all our trading strategies are universally measurable without invoking the axiom of projective determinacy.\\
Furthermore, we hope that by constructing an explicit arbitrage strategy in the proof of $\textit{(3)}\Rightarrow\textit{(1)}$ we can clarify the proof of \citep{burzoni2016universal}, Theorem 4.23, pp. 42-46 (in particular \citep{burzoni2016universal}[A.3]) by offering a similar to the above (but much simpler) reasoning for the case $\mathfrak{P}=\{\P\}$. Introducing a measurable separator $\xi$ it is apparent that $j_z$ in \citep[p.44]{burzoni2016universal} can always be chosen equal to one in our setting. Also the resulting strategy $H^{\P}$ therein can be chosen universally measurable.
\end{Rem}

To prove the first part of Theorem \ref{Thm. bn vs bfhmo} for the case $\Phi \neq 0$ we recall the following notion from \citep{bfhmo}:

\begin{Defn}[\citep{bfhmo}, Def. 4] \label{pahtspace}
A pathspace partition scheme $\mathcal{R}(\alpha^*,H^*)$ of $\Omega$ is a collection of trading strategies $H_1,...,H_{\beta} \in \mathcal{H}(\mathbb{F}^{\mathcal{U}})$, $\alpha_1, \dots, \alpha_{\beta}\in \R^k$ and arbitrage aggregators $\tilde{H}_0, \dots ,\tilde{H}_{\beta}$ for some $1 \le \beta \le k$ such that
\begin{enumerate}
\item the vectors $\alpha_i$, $1 \le i \le \beta$ are linearly independent,
\item for any $i \le \beta$
\begin{align*}
\alpha_i\cdot \Phi+H_i\circ S_T \ge 0 \hspace{0.5cm}\text{on }A_{i-1}^*,
\end{align*}
where $A_0= \Omega$, $A_i\coloneqq \{ \alpha_i\cdot \Phi+ H_i \circ S_T =0 \} \cap A_{i-1}^*$,
\item for any $i=0, \dots, \beta$, $\tilde{H}_i$ is an Arbitrage Aggregator for $A_i$,
\item if $\beta < k$, then either $A_{\beta}= \emptyset$ or for any $\alpha \in \R^k$ linearly independent from $\alpha_1, \dots, \alpha_{\beta}$ there does not exist $H$ such that
\begin{align*}
\alpha \cdot \Phi + (H \circ S_T) \ge 0 \hspace{0.5cm} \text{on } A_{\beta}^*.
\end{align*}
\end{enumerate}
\end{Defn}

\begin{Defn}[\citep{bfhmo}, Def. 5]
A pathspace partition scheme $\mathcal{R}(\alpha^*, H^*)$ is successful if $A_{\beta}^* \neq \emptyset$.
\end{Defn}

We quote the following results:

\begin{Lem}[\citep{bfhmo}, Lemma 5]
For any $\mathcal{R}(\alpha^*, H^*)$, $A_i^*= \Omega^*_{\{\alpha^j \cdot \Phi \ | \ j \le i\}}$. Moreover, if $\mathcal{R}(\alpha^*, H^*)$ is successful, then $A^*_{\beta}=\Omega_{\Phi}^*$.
\end{Lem}

\begin{Lem}[\citep{bfhmo}, Proof of Theorem 1 for $\Phi \neq 0$]
A pathspace partition scheme $\mathcal{R}(\alpha^*, H^*)$ is successful if and only if $\Omega_{\Phi}^*\neq \emptyset$.
\end{Lem}

We now complete the first part of the proof of Theorem \ref{Thm. bn vs bfhmo} for the case $\Phi \neq 0$:


\begin{proof}[Proof of Theorem \ref{Thm. bn vs bfhmo} for $\Phi\neq 0$]
The existence of $\Omega^{\mathfrak{P}}$ and $\textit{(1)} \Rightarrow$ No Strong Arbitrage in $\Ac_{\Phi}(\tilde{\mathbb{F}})$ on $\Omega^{\mathfrak{P}}$, $\textit{(2)}\Rightarrow\textit{(3)}$ follow exactly as before. We now argue that $\textit{(1)}\Rightarrow\textit{(2)}$ holds in the spirit of \citep[Theorem 5.1, p. 850]{bouchard2015arbitrage}, by induction over the number $e$ of options available for static trading. In particular we can assume without loss of generality that there exists a random variable $\varphi \ge 1$ such that $|\phi_j| \le \varphi$ for all $j=1, \dots, k$ and consider the set $\mathcal{Q}_{\varphi}= \{ \Q \in \mathcal{Q}_{\Pfr} \ | \ \E_{\Q}[\varphi]< \infty \}$ in order to avoid integrability issues. So let us assume there are $e \ge 0$ traded options $\phi_1, \dots, \phi_e$, for which $\textit{(1)}\Rightarrow \textit{(2)}$ holds. We introduce an additional option $g=\phi_{e+1}$ and assume $\P\left((\Omega^{\mathfrak{P}})^*_{\{\phi_1, \dots, \phi_{e+1}\}}\right)=1$ for all $\P \in \mathfrak{P}$. Then clearly $\P\left((\Omega^{\mathfrak{P}})^*_{\{\phi_1, \dots, \phi_{e}\}}\right)=1$ for all $\P \in \mathfrak{P}$ and by the induction hypothesis there is no arbitrage in the market with options $\{\phi_1, \dots, \phi_{e}\}$ available for static trading. Let $\P \in \mathfrak{P}$. Then by exactly the same arguments as in \citep[proof of Theorem 5.1(a)]{bouchard2015arbitrage} we can use convexity of $\mathcal{Q}_{\varphi}$ and Theorem \ref{Thm super} to find a measure $\Q \in \mathcal{Q}_{\varphi}$, such that $\P \ll \Q$ and $\Q\in \mathcal{Q}_{\Pfr, \{\phi_1,\dots, \phi_{e+1}\}}$, so $\textit{(2)}$ holds.\\
Lastly it remains to show $\textit{(3)}\Rightarrow\textit{(1)}$. Let us thus assume there exists $\hat{\P} \in \mathfrak{P}$ such that $\hat{\P}(\Omega^{\mathfrak{P}} \setminus (\Omega^{\mathfrak{P}})^*_{\Phi})>0$. We want to find $(h,H) \in \admU$ and $\tilde{\P}\in \mathfrak{P}$ such that $h \cdot \Phi + H \circ S_T \ge 0$ $\mathfrak{P}$-q.s and $\tilde{\P}(h \cdot \Phi+H \circ S_T>0)>0$.
We use the properties of a pathspace partition scheme $\mathcal{R}(\alpha^*, H^*)$ recalled above. We define
\begin{align*}
m &= \min (k \in \{0, \dots , \beta\} \ | \  \tilde{\P}(A_k \setminus A_{k}^*)>0 \text{ for some }\tilde{\P} \in \mathfrak{P}) \\
\tilde{m} &= \min (k \in \{1, \dots , \beta\} \ | \  \tilde{\P}(A_{k-1}^* \setminus A_{k})>0 \text{ for some }\tilde{\P} \in \mathfrak{P}),
\end{align*}
where $A_0= \Omega^{\mathfrak{P}}$. If $\tilde{m}\le m$ then we select the strategy $(\alpha_{\tilde{m}},H_{\tilde{m}}) \in \admU$ which satisfies $H_{\tilde{m}} \circ S_T + \alpha_{\tilde{m}} \cdot \Phi \ge 0$ on $A_{\tilde{m}-1}^*$. We note that $\P(A_{\tilde{m}-1}^*)=1$ for all $\P \in \mathfrak{P}$ by definition of $m, \ \tilde{m}$ and $\{ H_{\tilde{m}} \circ S_T + \alpha_{\tilde{m}} \cdot \Phi > 0\}=A_{\tilde{m}-1}^* \setminus A_{\tilde{m}}$, so that $\tilde{\P}(H_{\tilde{m}} \circ S_T + \alpha_{\tilde{m}} \cdot \Phi > 0)>0$ for some $\tilde{\P}\in \mathfrak{P}$. If $\tilde{m}>m$, then $\P(A_m)=1$ for all $\P \in \mathfrak{P}$, $\tilde{\P}(A_m \setminus A^*_m)>0$ for some $\tilde{\P}\in \mathfrak{P}$, so we can argue as in the proof of Proposition \ref{Thm. bn vs bfhmo} for $\Phi=0$ $\textit{(3)} \Rightarrow\textit{(1)}$ using a standard separator and measurable selection of a measure in $\mathfrak{P}$.\\ As before, the second part of Theorem \ref{Thm. bn vs bfhmo} follows immediately from Proposition \ref{prop 2}. This concludes the proof.
\end{proof}

\begin{proof}[Proof of Theorem \ref{Thm. S}]
We recall the analytic set $\Omega^{\mathfrak{P}}$ from the proof of Theorem \ref{Thm. bn vs bfhmo} for $\Phi=0$ and the sets $\{C_n\}_{n \in \N}$ from \eqref{eq. app}. Now we define 
\begin{align*}
B:= \bigcup_{n \in \N} \{C_n \ | \ \P( C_n \cap (\Omega^{\mathfrak{P}})^*_{\Phi}) =0 \text{ for all }\P \in \mathfrak{P}\} \in \mathcal{B}(X).
\end{align*}
We claim that \eqref{eq. app} implies that
\begin{align*}
B \cap  (\Omega^{\mathfrak{P}})^*_{\Phi} &= \bigcup \{C \in \mathcal{S} \ | \ \P( C \cap (\Omega^{\mathfrak{P}})^*_{\Phi})=0 \text{ for all }\P \in \mathfrak{P}\}\cap (\Omega^{\mathfrak{P}})^*_{\Phi}\\
&=\bigcup\{C\in \mathcal{S} \ | \ C \cap (\Omega^{\mathfrak{P}})^*_{\Phi} \in \mathcal{N}^{\mathfrak{P}} \}\cap (\Omega^{\mathfrak{P}})^*_{\Phi}.
\end{align*}
Indeed, clearly $B \subseteq \bigcup\{C\in \mathcal{S} \ | \  C \cap (\Omega^{\mathfrak{P}})^*_{\Phi} \in \mathcal{N}^{\mathfrak{P}} \}$. 
Now assume towards a contradiction that there exists $$\omega \in \left(\bigcup\{C\in \mathcal{S} \ | \  C \cap (\Omega^{\mathfrak{P}})^*_{\Phi} \in \mathcal{N}^{\mathfrak{P}} \}\cap (\Omega^{\mathfrak{P}})^*_{\Phi}\right)\setminus (B\cap (\Omega^{\mathfrak{P}})^*_{\Phi}).$$ In particular $\omega \in C$ for some $C \in \mathcal{S}$ such that $C \cap (\Omega^{\mathfrak{P}})^*_{\Phi} \in \mathcal{N}^{\mathfrak{P}}$. By By \eqref{eq. app} there exists $n_0 \in \N$ such that $C_{n_0} \subseteq C$ and $\omega \in C_{n_0}$. This implies $\omega \in B$ and thus shows the claim.\\
Let us now first assume that $\Phi=0$ and set
\begin{align}\label{eq:star}
\Omega \coloneqq \Omega^{\mathfrak{P}} \setminus((\Omega^{\mathfrak{P}})^* \cap B) \in \mathcal{F}^{\mathcal{U}}.
\end{align}
By assumption we have $\P(\Omega)=1$ for all $\P \in \mathfrak{P}$. By definition of the $()^*$ operation
\begin{align*}
\Omega^{*}=((\Omega^{\mathfrak{P}})^* \setminus B)_{\Phi}^*=(\Omega^{\mathfrak{P}} \setminus B)^*
\end{align*}
follows. To see the above equality, take a martingale measure $\Q \in \mathcal{M}_{\Omega, \Phi}$ and assume that $\Q(\Omega \setminus (\Omega^{\mathfrak{P}} \setminus B)) >0$. As $\Omega \setminus (\Omega^{\mathfrak{P}} \setminus B)= \Omega^{\mathfrak{P}} \setminus (\Omega^{\mathfrak{P}})^*\cap B$ 
we conclude that $\Q(\Omega^{\mathfrak{P}} \setminus (\Omega^{\mathfrak{P}})^*)>0$. Since any calibrated martingale measure supported on a subset of $\Omega$ is in $\mathcal{M}_{\Omega^{\mathfrak{P}}}$ this leads to a contradiction to the definition of $(\Omega^{\mathfrak{P}})^*$. Also, $\Omega^{\mathfrak{P}} \setminus B=\Omega^{\mathfrak{P}} \cap B^c$ is the intersection of two analytic sets, so we conclude that $\Omega^*$ is analytic. Lastly, by definition of $\Omega^{\mathfrak{P}}$ we conclude $\Omega^{*}=(\Omega^{\mathfrak{P}})^*$ $\mathfrak{P}$-q.s..

The implications \textit{(1)} $\Rightarrow$ \textit{(2)} $\Rightarrow$ \textit{(3)} $\Rightarrow$ \textit{(4)} $\Rightarrow$ \textit{(5)} follow directly from the definition. Thus we only need to show \textit{(5)} $\Rightarrow$ \textit{(1)}. Let us fix $C \in \mathcal{S}$ such that $C \subseteq \Omega$. No Arbitrage de la Classe $\mathcal{S}$ on $\Omega$ implies that $\Omega^* \cap C\neq \emptyset$. 
From \eqref{eq:star} we thus conclude that $\P((\Omega^{\mathfrak{P}})^*\cap C)>0$ for some $\P \in \mathfrak{P}$. As $\Omega^*=(\Omega^{\mathfrak{P}})^*$ $\mathfrak{P}$-q.s. this implies $\P(\Omega^* \cap C)>0$.
Using a construction similar to the proof of Proposition \ref{Thm. bn vs bfhmo} for the case $\Phi=0$, we can find a measure $\tilde{\P} \in \mathfrak{P}$ such that $\tilde{\P}(C)>0$ and  0 is in the interior of the conditional support of $\tilde{\P}(\cdot |\Omega^{*})$. By  \citep[Theorem 1]{rokhlin2008proof}, we conclude that there exists a martingale measure $\Q \in \mathcal{Q}_{\Pfr}$ equivalent to $\tilde{\P}(\cdot|\Omega^{*})$, in particular $\Q(C)>0$. The case $\Phi \neq 0$ can now be treated similarly: indeed, we define $\Omega^{\mathfrak{P},\Phi}$ as in the proof of Theorem \ref{Thm. bn vs bfhmo} for $\Phi=0$, but now including the statically traded options $\Phi$ in the definition of the quasi-sure support and follow the same arguments as above. This concludes the proof.
\end{proof}

\subsection{Proof of Theorem \ref{Thm super}} \label{Sec proofsuper}
We first show that the quasi-sure superhedging theorem of \citep{bouchard2015arbitrage} implies the second part of Theorem \ref{Thm super}.

\begin{Prop}\label{prop sup}
Let $\Omega$ be an analytic subset of $X$ and $\Omega^*_{\Phi} \neq \emptyset$. Let the set $\mathfrak{P}$ satisfy (APS) and $\Nc^\Pfr=\Nc^{\mathcal{M}_{\Omega, \Phi}^f}$ 
Then $\Nc^{\mathcal{Q}_{\Pfr,\Phi}}=\Nc^{\mathcal{M}_{\Omega, \Phi}^f}$ and for an upper semianalytic function $g: X \to \R$ 
\begin{align}\label{eq. duality}
\sup_{\Q \in \mathcal{M}^f_{\Omega, \Phi}} \E_{\Q}[g]
&= \pi_{\Omega_{\Phi}^*}(g) = \pi^{\mathfrak{P}}(g)=\sup_{\Q \in \mathcal{Q}_{\Pfr,\Phi}}\E_{\Q}[g].
\end{align}
\end{Prop}
\begin{proof}
That $\mathcal{Q}_{\Pfr, \Phi}$ and $\mathcal{M}_{\Omega,\Phi}^f$ have the same polar sets follows by the definition of $\Omega^*_{\Phi}$ and \cite[Lemma 2]{bfhmo}. We now show \eqref{eq. duality}: consider
\begin{align*}
\mathfrak{P}^{\Omega}\coloneqq\mathcal{P}^f(\Omega^*_{\Phi}).
\end{align*}
Note that there is no $\mathcal{M}^f_{\Omega, \Phi}$-q.s. arbitrage iff there is no $\mathfrak{P}^{\Omega}$-q.s arbitrage.\\
We now show that $\Omega^*_{\Phi}$ is analytic if $\Omega$ is analytic.
Recall the set $\mathcal{P}_{Z, \Phi}$ from Lemma 5.4 of \citep{burzoni2015model}, page 13 defined by
\begin{align*}
\mathcal{P}_{Z,\Phi}\coloneqq \left\{ \P \in \mathcal{P}^f(X) \ | \ \exists \Q \in \mathcal{M}^f_{X,\Phi} \text{ such that } \frac{d\Q}{d\P}=\frac{c(\P)}{1+Z}\right\},
\end{align*}
where $Z= \max_{i=1, \dots, d} \max_{t=0, \dots, T} S_t^i$ and $c(\P)=(\E_{\P}[1+Z]^{-1})^{-1}$. \citep{burzoni2015model} show that the set 
\begin{align*}
\{(\omega, \P) \ | \ \omega \in X^*, \ \P \in \mathcal{P}^{\omega} \}
\end{align*}
is analytic, where $\mathcal{P}^{\omega}= \{ \P \in \mathcal{P}_{Z,\Phi} \ | \ \P(\{\omega \})>0 \}$. Note that
\begin{align*}
\{(\omega, \P) \ | \ \omega \in X^*, \ \P \in \mathcal{P}^{\omega} \} \cap \left(\Omega \times \mathcal{P}^f(\Omega) \right)
\end{align*}
is analytic and the projection of the above set to the first coordinate is exactly $\Omega^*_{\Phi}$, which shows that $\Omega^*_{\Phi}$ is analytic. 
We note $\omega \mapsto \mathfrak{P}^{\Omega}_{t}(\omega)= \mathcal{P}^f(\text{proj}_{t+1}(\Sigma_t^{\omega} \cap \Omega^*_{\Phi}))$ has analytic graph by exactly the same argument as in the proof of Proposition \ref{prop 2} replacing $\Omega$ by $\Omega^*_{\Phi}$. The result now follows from the Superhedging Theorem of \citep{bouchard2015arbitrage} and the definition of $\mathcal{M}_{\Omega,\Phi}^f$. 
\end{proof}

We now show that the classical $\P$-a.s. one-step superhedging duality can be deduced by means of pathwise reasoning:
\begin{Lem}\label{lem. p-superhedging}
Let $t\in \{0, \dots, T-1\}$ and $g: X_{t+1} \to \R$ be $\mathcal{F}^{\mathcal{U}}_{t+1}$-measurable. Let $\P \in \mathcal{P}(X_1)$ and fix $\omega \in X_t$ such that \textbf{NA}$(\P)$ holds for the one-period model $(S_t(\omega), S_{t+1}(\omega,\cdot))$. Then
\begin{align*}
\sup_{\Q \sim \P, \ \Q \in \mathcal{M}_{X_1}} \E_{\Q}[g(\omega, \cdot)]= \inf\{x \in \R \ | \ \exists H \in \R^d \text{ s.t. } x+H\Delta S_{t+1}(\omega, \cdot) \ge g(\omega, \cdot) \ \P\text{-a.s.}\}.
\end{align*}
\end{Lem}

\begin{proof}
As $g$ is $\mathcal{F}^{\mathcal{U}}_{t+1}$-measurable, by \citep[Lemma 7.27, p.173]{bertsekas1978stochastic} there exists a Borel-measurable function $\tilde{g}:(\R^d)^{t+1} \to \R$ such that $g(\omega)=\tilde{g}(S_{0:t+1}(\omega))$ for $\P$-a.e. $\omega \in X_{t+1}$. Assume first that $S_{t+1}\mapsto \tilde{g}(S_{0:t}(\omega),S_{t+1})$ is continuous. Define $\chi^{\P}\coloneqq\text{supp}(\P \circ \Delta S_{t+1}(\omega, \cdot)^{-1})$. Then as \textbf{NA}$(\P)$ holds $\chi^{\P}=(\chi^{\P})^*$ and thus by \citep[Theorem 2]{bfhmo} and continuity of $S_{t+1}\mapsto \tilde{g}(S_{0:t}(\omega),S_{t+1})$ as well as $S_{t+1} \mapsto H(S_{t+1}-S_t(\omega))$
\begin{align*}
\sup_{\Q \sim \P, \ \Q \in \mathcal{M}_{X_1}} \E_{\Q}[g(\omega, \cdot)]&\le \inf\{x \in \R \ | \ \exists H \in \R^d \text{ s.t. } x+H\Delta S_{t+1}(\omega, \cdot) \ge g(\omega, \cdot) \ \P\text{-a.s.}\}  \\
&=\inf\{x \in \R \ | \ \exists H \in \R^d \text{ s.t. } x+H\Delta S_{t+1}(\omega, \cdot) \ge \tilde{g}(S_{0:t}(\omega), \cdot) \ \text{on } \chi^{\P}\}\\
&= \sup_{\Q \in \mathcal{M}^f_{\chi^{\P}}} \E_{\Q}[\tilde{g}(S_{0:t}(\omega), \cdot)]\\
&\le \sup_{\Q \sim \P\circ\Delta S_{t+1}(\omega, \cdot)^{-1}, \ \Q \in \mathcal{M}_{\R^d}} \E_{\Q}[\tilde{g}(S_{0:t}(\omega), \cdot)]\\
&=\sup_{\Q \sim \P, \ \Q \in \mathcal{M}_{X_1}} \E_{\Q}[g(\omega, \cdot)].
\end{align*}
If $S_{t+1}\mapsto \tilde{g}(S_{0:t}(\omega),S_{t+1})$ is Borel-measurable, then by Lusin's theorem (see \citep[Theorem 7.4.3, p.227]{cohn2013measure}) there exists an increasing sequence of compact sets $(K_n)_{n\in \N}$ such that  $K_n \subseteq \chi^{\P}$, $\P \circ \Delta S_{t+1}(\omega, \cdot)^{-1}(K_n^c)\le 1/n$ and $\tilde{g}(S_{0:t}(\omega), \cdot)|_{K_n}$ is continuous. In particular there exists $n_0 \in \N$ such that for all $n\ge n_0$ we have $K_n=(K_n)^*$. By the above argument
\begin{align}\label{eq. jw}
&\inf\{x \in \R \ | \ \exists H \in \R^d \text{ s.t. } x+H\Delta S_{t+1}(\omega, \cdot) \ge \tilde{g}(S_{0:t}(\omega), \cdot) \ \text{on } K_n\}\\
&\qquad = \sup_{\Q \sim \P\circ\Delta S_{t+1}(\omega, \cdot)^{-1} \, \ \Q \in \mathcal{M}_{K_n}} \E_{\Q}[g(\omega, \cdot)] \nonumber
\end{align}
holds for $n \ge n_0$. The claim now follows by taking suprema in $n \in \N$ on both sides of \eqref{eq. jw}:
\begin{align*}
&\inf\{x \in \R \ | \ \exists H \in \R^d \text{ s.t. } x+H\Delta S_{t+1}(\omega, \cdot) \ge g(\omega, \cdot) \ \P\text{-a.s.}\}\\
= \ &\sup_{n \in \N}\inf\{x \in \R \ | \ \exists H \in \R^d \text{ s.t. } x+H\Delta S_{t+1}(\omega, \cdot) \ge \tilde{g}(S_{0:t}(\omega), \cdot) \ \text{on } K_n\} \\
=\ &\sup_{n \in \N}\sup_{\Q \sim \P\circ\Delta S_{t+1}(\omega, \cdot)^{-1} \, \ \Q \in \mathcal{M}_{K_n}} \E_{\Q}[\tilde{g}(S_{0:t}(\omega), \cdot)] \\
\le \ &\sup_{\Q \sim \P\circ\Delta S_{t+1}(\omega, \cdot)^{-1} \, \ \Q \in \mathcal{M}_{\R^d}} \E_{\Q}[\tilde{g}(S_{0:t}(\omega), \cdot)]\\
\le\ &\inf\{x \in \R \ | \ \exists H \in \R^d \text{ s.t. } x+H\Delta S_{t+1}(\omega, \cdot) \ge g(\omega, \cdot) \ \P\text{-a.s.}\}.
\end{align*}
\end{proof}

Using this one-step duality result under fixed $\P$ and (APS) of $\mathfrak{P}$ we now prove the first part of Theorem \ref{Thm super}, which is restated in the following proposition:
\begin{Prop}\label{prop. joh1}
Let \textbf{NA}$(\mathfrak{P})$ hold and let $g: X \to \R$ be upper semianalytic. Then there exists a measure $\P^g= \P^g_0 \otimes \cdots \otimes \P^g_{T-1}$ and an $\mathcal{F}^{\mathcal{U}}$-measurable set $\Omega_g^{\mathfrak{P}}$ with $\P(\Omega^{\mathfrak{P}}_g)=1$ for all $\P \in \mathfrak{P}$, such that
\begin{align}
\pi^{\mathfrak{P}}(g)=\pi^{\hat{\P}}(g)=\pi_{(\Omega_g^{\mathfrak{P}})_{\Phi}^*}(g)= \sup_{\Q \in \mathcal{M}_{\Omega^{\mathfrak{P}}_g,\Phi}} \E_{\Q}[g]=\sup_{\Q \in \mathcal{Q}_{\Pfr,\Phi}}\E_{\Q}[g].\nonumber
\end{align}
\end{Prop}

\begin{proof}
We note that by \textbf{NA}$(\mathfrak{P})$ and Theorem \ref{Thm. bn vs bfhmo} the difference $\Omega^{\mathfrak{P}}\setminus (\Omega^{\mathfrak{P}})^*_{\Phi}$ is $\mathfrak{P}$-polar. We first take $\Phi=0$. Recall the definition of the one-step functionals given in \citep[Lemma 4.10, p. 846]{bouchard2015arbitrage} 
\begin{align*}
\mathcal{E}_{T}(g)(\omega)&=g(\omega)\\
\mathcal{E}_{t}(g)(\omega)&= \sup_{\Q \in \mathcal{Q}_{t}(\omega)} \E_{\Q}[\mathcal{E}_{t+1}(g)(\omega, \cdot)], \hspace{0.5cm} t=0, \dots, T-1.
\end{align*}
By (APS) and upper semianalyticity of g, every $\mathcal{E}_t(g)$ is upper semianalytic. We show recursively that for every $t=0, \dots, T-1$ and for $\mathfrak{P}$-q.e. $\omega \in X_t$ there exists a measure $\P \in \mathcal{P}(X_1)$ such that \textbf{NA}$(\P)$ holds and
\begin{align*}
\sup_{\Q \in \mathcal{Q}_t(\omega)} \E_{\Q}[\mathcal{E}_{t+1}(g)(\omega, \cdot)]=\sup_{\Q \sim \P, \ \Q \in \mathcal{M}_{X_1}} \E_{\Q}[\mathcal{E}_{t+1}(g)(\omega, \cdot)].
\end{align*}
Note that by measurable selection arguments and construction of $\Omega^{\mathfrak{P}}$ we conclude that for $\mathfrak{P}$-q.e. $\omega \in X_t$ the properties \textbf{NA}$(\mathfrak{P}_t(\omega))$ and $\P(\text{proj}_{t+1}(\Omega^{\mathfrak{P}} \cap \Sigma_{t}^{\omega}))=1$ hold for all $\P \in \mathfrak{P}_{t}(\omega)$. We now fix $t \in \{0, \dots, T-1\}$ and $\omega \in X_t$ such that \textbf{NA}$(\mathfrak{P}_t(\omega))$ and $\P(\text{proj}_{t+1}(\Omega^{\mathfrak{P}} \cap \Sigma_{t}^{\omega}))=1$ for all $\P \in \mathfrak{P}_{t}(\omega)$ holds. Note that there exists a sequence $(\P_n)_{n \in \N}$ such that $\P_n \in \mathfrak{P}_t(\omega)$ for all $n \in \N$ and 
\begin{align*}
\sup_{\Q \ll \P_n, \ \Q \in \mathcal{M}_{X_1}} \E_{\Q}[\mathcal{E}_{t+1}(g)(\omega, \cdot)] \uparrow \sup_{\Q \in \mathcal{Q}_t(\omega)} \E_{\Q}[\mathcal{E}_{t+1}(g)(\omega, \cdot)] \quad (n \to \infty).
\end{align*}
We see from the proof of Theorem \ref{Thm. bn vs bfhmo} for $\Phi=0$ in Section \ref{subseq. ftap} that under \textbf{NA}($\mathfrak{P}_t(\omega)$) and for a fixed $\P \in \mathfrak{P}_t(\omega)$, we can always find $\tilde{\P} \in \mathfrak{P}_t(\omega)$ such that $\P \ll \tilde{\P}$ and \textbf{NA}($\tilde{\P}$) holds. Thus we can assume without loss of generality that \textbf{NA}$(\P_n)$ holds for all $n \in \N$. Define $\hat{\P}_n:= \sum_{k=1}^{n}2^{-k}/(1-2^{-n})\P_k \in \mathfrak{P}_t(\omega)$ as well as $\P^g_t(\omega):= \sum_{k=1}^{\infty}2^{-k} \P_k$ and note that \textbf{NA}$(\hat{\P}_n)$ as well as \textbf{NA}$(\P^g_t(\omega))$ hold for all $n \in \N$. Furthermore 
\begin{align*}
\mathcal{E}_t(g)(\omega)&=\sup_{n \in \N}\sup_{\Q \ll \P_n, \ \Q \in \mathcal{M}_{X_1}} \E_{\Q}[\mathcal{E}_{t+1}(g)(\omega, \cdot)] \le \sup_{n \in \N}\sup_{\Q \sim \hat{\P}_n, \ \Q \in \mathcal{M}_{X_1}} \E_{\Q}[\mathcal{E}_{t+1}(g)(\omega, \cdot)]\\
&\le\sup_{\Q \sim \P^g_t(\omega), \ \Q \in \mathcal{M}_{X_1}} \E_{\Q}[\mathcal{E}_{t+1}(g)(\omega, \cdot)]
\\
&=\inf \{ x \in \R \ | \ \exists H\in\R^d \text{ s.t. } x + H\Delta S_{t+1}(\omega, \cdot) \ge \mathcal{E}_{t+1}(g)(\omega, \cdot)) \ \P^g_t(\omega)\text{-a.s.}\},
\end{align*}
where the last equality follows from Lemma \ref{lem. p-superhedging}. Define
\begin{align*}
\pi^{\P^g_t(\omega)}_t(g)&=\inf \{ x \in \R \ | \ \exists H\in\R^d \text{ s.t. } x + H\Delta S_{t+1}(\omega, \cdot) \ge \mathcal{E}_{t+1}(g)(\omega, \cdot) \ \P^g_t(\omega)\text{-a.s.}\}\\
\pi^{\mathfrak{P}_t(\omega)}(g)&=\inf \{ x \in \R \ | \ \exists H\in\R^d \text{ s.t. } x + H\Delta S_{t+1}(\omega, \cdot) \ge \mathcal{E}_{t+1}(g)(\omega, \cdot) \ \mathfrak{P}_t(\omega)\text{-q.s.}\}.
\end{align*}
Clearly $\pi^{\P^g_t(\omega)}(g) \le \pi^{\mathfrak{P}_t(\omega)}(g)$. Now assume towards a contradiction that the inequality is strict and set $\epsilon:=\pi^{\mathfrak{P}_t(\omega)}(g)-\pi^{\P^g_t(\omega)}(g)>0$. Furthermore note that for a sequence of compact sets $(K_n)_{n \in \N}$ such that $K_n \uparrow \R^d$ we have
\begin{align*}
&\pi^{\P^g_t(\omega)|_{K_n}}(g):= \inf \{ x \in \R \ | \ \exists H\in\R^d \text{ s.t. } x + H\Delta S_{t+1} \ge \mathcal{E}_{t+1}(g) \ \P^g_{t}(\omega)(\cdot|\Delta S_{t+1}(\omega, \cdot) \in K_n)\text{-a.s.}\}\\
&\qquad \qquad\uparrow \pi^{\P^g_t(\omega)}(g) \quad (n \to \infty),\\
&\pi^{\mathfrak{P}_t(\omega)|_{K_n}}(g):=\inf \{ x \in \R \ | \ \exists H\in\R^d \text{ s.t. } x + H\Delta S_{t+1} \ge \mathcal{E}_{t+1}(g)\ \mathfrak{P}_t(\omega)|_{K_n}\text{-q.s.}\}\\
&\qquad \qquad \uparrow \pi^{\mathfrak{P}_t(\omega)}(g) \quad (n \to \infty),
\end{align*}
where $$\mathfrak{P}_t(\omega)|_{K_n}:= \{ \P(\cdot| \Delta S_{t+1}(\omega, \cdot) \in K_n) \ | \ \P \in \mathfrak{P}_t(\omega), \ \P(\Delta S_{t+1}(\omega, \cdot) \in K_n) >0\}.$$
Choose $n \in \N$ large enough, such that $\pi^{\mathfrak{P}_t(\omega)|_{K_n}}(g)- \pi^{\P^g_t(\omega)|_{K_n}}(g)>3\epsilon/4.$ Denote by $\mathcal{H}_{K_n}$ the closed set of $H \in \R^d$ such that
\begin{align*}
\pi^{\P^g_t(\omega)|_{K_n}}(g)+ \epsilon/2+H\Delta S_{t+1}(\omega, \cdot) \ge \mathcal{E}_{t+1}(g)(\omega, \cdot)) \quad \P^g_t(\omega)(\cdot|\Delta S_{t+1}(\omega, \cdot) \in K_n)\text{-a.s.}
\end{align*}
Then for every $H \in \mathcal{H}_{K_n}$ there exists $\P_n^H \in \mathfrak{P}_t(\omega)$ such that 
\begin{align*}
\P_n^H(\{\pi^{\P^g_t(\omega)|_{K_n}}(g)+\epsilon/2+H\Delta S_{t+1}(\omega, \cdot) < \mathcal{E}_{t+1}(g)(\omega, \cdot))\} \cap \{\Delta S_{t+1}(\omega, \cdot) \in K_n\} )>0.
\end{align*}
Note that there exists a countable sequence $(H^k_n)_{k \in \N}$, which is dense in $\mathcal{H}_{K_n}$. In particular for every $H \in \R^d$ such that
\begin{align*}
\pi^{\P^g_t(\omega)|_{K_n}}(g)+ \epsilon/4+H\Delta S_{t+1}(\omega, \cdot) \ge \mathcal{E}_{t+1}(g)(\omega, \cdot)) \ \P^g_t(\omega)(\cdot|\Delta S_{t+1}(\omega, \cdot) \in K_n)\text{-a.s.}
\end{align*}
there exists $k \in \N$ such that
\begin{align*}
\pi^{\P^g_t(\omega)|_{K_n}}(g)+ \epsilon/2+H^k_n\Delta S_{t+1}(\omega, \cdot) \ge \mathcal{E}_{t+1}(g)(\omega, \cdot)) \ \P^g_t(\omega)(\cdot|\Delta S_{t+1}(\omega, \cdot) \in K_n)\text{-a.s.}
\end{align*}
Set now $\P^n=\sum_{k=1}^{\infty}2^{-k}\P^{H^k_n}_n \in \mathcal{P}(X_1)$ and note that for all $n \in \N$ large enough
\begin{align*}
&\pi^{\frac{1}{2}(\P^g_t(\omega)+\P^n)|_{K_n}}(g)-\pi^{\P^g_t(\omega)|_{K_n}}(g)\ge \epsilon/4.
\end{align*}
Taking $K_n \uparrow \R^d$ we have in particular
\begin{align*}
\mathcal{E}_t(g)(\omega)\le\sup_{\Q \sim \P^g_t(\omega), \ \Q \in \mathcal{M}_{X_1}} \E_{\Q}[g]<\lim_{n \to \infty}\sup_{\Q \sim \frac{1}{2}(\P^g_t(\omega)+\P^n)|_{K_n}, \ \Q \in \mathcal{M}_{X_1}} \E_{\Q}[g] \le \mathcal{E}_t(g)(\omega),
\end{align*}
a contradiction. Thus 
\begin{align*}
\mathcal{E}_t(g)(\omega)&=\pi^{\mathfrak{P}_t(\omega)}(g)=\pi^{\P^g_t(\omega)}(g).
\end{align*}
As $0 \in \text{ri}(\text{supp}((\P^g_t(\omega)))$ a natural universally measurable candidate for a superhedging strategy $\omega \mapsto H_{t+1}(\omega)$ is the right derivative $\lim_{\epsilon \in \Q,\ \epsilon \downarrow 0}(\mathcal{E}^{\epsilon e_i}_t(g)(\omega)-\mathcal{E}_t(g)(\omega))/\epsilon$ where $\mathcal{E}^{\epsilon e_i}_t(g)(\omega)$ is the superhedging price for the Borel-measurable stock $(S_t+\epsilon e_i)$, $i=1, \dots, d' \le d$ instead of $S_t$. This is a pointwise limit of differences of upper seminanalytic functions and thus universally measurable. For $\omega \in X_t$ such that this quantity does not exist, we set $H_{t+1}(\omega)=0$.  Furthermore in order to show that the map $\omega \mapsto \P^g_t(\omega)$ can be chosen to be universally measurable we first note that in \citep[Lemma 4.8, p.843]{bouchard2015arbitrage} the set 
\begin{align*}
\{ (\Q, \P) \in \mathcal{P}(\text{proj}_{t+1}(\Sigma_t^{\omega})) &\times \mathcal{P}(\text{proj}_{t+1}(\Sigma_t^{\omega})) \ | \ \\&\E_{\Q}[\Delta S_{t+1}(\omega, \cdot)]=0, \ \P \in \mathfrak{P}_t(\omega), \ \Q \ll \P \}
\end{align*}
is analytic. Thus we can apply the Jankov-von-Neumann selection theorem (see \citep[Proposition 7.50, p.184]{bertsekas1978stochastic}) to find $1/n$-optimisers $(\Q_t^n(\omega), \P_t^n(\omega))$ for $\mathcal{E}_t(g)(\omega)$ and the claim follows. The case $\Phi\neq 0$ can be handled by induction as in the proof of Theorem \ref{Thm. bn vs bfhmo} for $\Phi\neq 0$. \\
In conclusion we have found a strategy $(h,H) \in \admU$ such that 
\begin{align*}
\sup_{\Q \in \mathcal{Q}_{\Pfr,\Phi}} \E_{\Q}[g]+h \cdot \Phi+(H \circ S_T) \ge g \hspace{0.5cm} \mathfrak{P}\text{-q.s.}
\end{align*}
We now define 
\begin{align*}
\Omega^{\mathfrak{P}}_{g}=\Omega^{\mathfrak{P}} \cap \left\{\omega \in X \ \Bigg| \ \sup_{\Q \in \mathcal{Q}_{\Pfr,\Phi}} \E_{\Q}[g]+h \cdot \Phi(\omega)+(H \circ S_T)(\omega) \ge g(\omega)\right\} \in \mathcal{F}^{\mathcal{U}}.
\end{align*}
This concludes the proof.
\end{proof}

\begin{Rem} 
By \textbf{NA}$(\mathfrak{P})$ Proposition \ref{prop. joh1} implies for $g=0$
\begin{align*}
0 &= \inf \{ x \in \R \ | \ \exists (h,H) \in \admU \text{ such that } x + h \cdot \Phi + (H \circ S_T) \ge 0 \ \mathfrak{P}\text{-q.s.}\}\\
&=\inf_{\tilde{g}\in \mathfrak{E}_0}\inf \{ x \in \R \ | \ \exists (h,H) \in \admU  \text{ such that } x+ h \cdot \Phi+(H \circ S_T) \ge \tilde{g} \ \text{ on } (\Omega_{\tilde g}^{\Pfr})^*_{\Phi} \}\\
&=\inf_{\tilde{g} \in \mathfrak{E}_0}\sup_{\Q \in \mathcal{M}_{\Omega^{\Pfr}_{\tilde g},\Phi}} \E_{\Q}[\tilde{g}],
\end{align*}
where we define
\begin{align*}
\mathfrak{E}_0 = \{ \tilde{g}: X \to (-\infty,0] \ \mathcal{F}^{\mathcal{U}}\text{-measurable} \ | \ \tilde{g}=0 \ \mathfrak{P}\text{-q.s.}   \}.
\end{align*}
In particular for every $\tilde{g} \in \mathfrak{E}_0$ there exists $\Q \in \mathcal{M}_{X,\Phi}$ such that $\E_{\Q}[\tilde{g}]=0$. A similar result was obtained by \citep{riedsonbuz17} in a more general setup. Aggregating the martingale measures corresponding to all $\tilde{g}$ (and thus to all $\mathfrak{P}$-polar sets) to achieve a result comparable to \citep{bouchard2015arbitrage} in a setup without using (APS) of $\mathfrak{P}$ remains an open problem.
\end{Rem}

\emph{Data Availability Statement:} Data sharing is not applicable to this article as no new data were created or analyzed in this study.

\makeatletter
\renewcommand{\@biblabel}[1]{[#1]}
\makeatother
\bibliographystyle{apalike}
\bibliography{bib}

\end{document}